\documentclass[journal,12pt,onecolumn,draftclsnofoot,]{IEEEtran}
\IEEEoverridecommandlockouts

\ifCLASSINFOpdf
  \usepackage[pdftex]{graphicx}
\else
\fi
\usepackage{cite}
\usepackage[cmex10]{amsmath}
\usepackage{array}
\usepackage{amsthm}
\usepackage{amsmath}
\usepackage{amsfonts}
\usepackage{balance}
\usepackage[usestackEOL]{stackengine}
\usepackage{amssymb}
\usepackage[font=small,labelfont=bf]{caption}
\newtheorem{definition}{Definition}
\newtheorem{theorem}{Theorem}
\newtheorem{corollary}{Corollary}
\newtheorem{lemma}{Lemma}

\newtheorem{remark}{Remark}

\newtheorem{assumption}{Assumption}[section]
\allowdisplaybreaks
\newcommand\norm[1]{\left\lVert#1\right\rVert}

\newcommand\myeqaa{\mathrel{\stackrel{\makebox[0pt]{\mbox{\normalfont\tiny (aa)}}}{=}}}
\newcommand\myeqab{\mathrel{\stackrel{\makebox[0pt]{\mbox{\normalfont\tiny (ab)}}}{\geq}}}
\newcommand\myeqac{\mathrel{\stackrel{\makebox[0pt]{\mbox{\normalfont\tiny (ac)}}}{=}}}

\begin{document}
%
\title{The Effect of Time Delay on the Average Data Rate and Performance in Networked Control Systems}
%
%
%

\author{Mohsen~Barforooshan,
	Milan~S.~Derpich,~\IEEEmembership{Member,~IEEE,}
	~Photios~A.~Stavrou,~\IEEEmembership{Member,~IEEE,}
	and~Jan~\O stergaard,~\IEEEmembership{Senior Member,~IEEE}
\thanks{Part of the results of this paper was presented at the 2017 American Control Conference \cite{barforooshan2017interplay} and the 56th IEEE Conference on Decision and Control\cite{barforooshan2017achievable}.

M. Barforooshan and J. \O stergaard are with the Department of
Electronic Systems, Aalborg University, DK-9220, Aalborg, Denmark (email: mob@es.aau.dk; jo@es.aau.dk). Their work has received funding from VILLUM FONDEN Young Investigator Programme, under grant agreement No. 19005 and partially from . 

M. S. Derpich is with the Department of Electronic Engineering, Universidad T\'ecnica Federico Santa Mar\'ia, Casilla 110-V, Valpara\'iso, Chile (email: milan.derpich@usm.cl). His work has received funding from FONDECYT project 1171059 and CONICYT Basal research grant FB0008. 

P. A. Stavrou is with the Department of Information Science and Engineering, KTH Royal Institute of Technology, Stockholm 100 44, Sweden (email: fstavrou@kth.se). His work has received funding from the Swedish Foundation for Strategic Research. 
}
}

\maketitle
\begin{abstract}
\noindent
This paper studies the performance of a feedback control loop closed via an error-free digital communication channel with transmission delay. The system comprises a discrete-time noisy linear time-invariant (LTI) plant whose single measurement output is mapped into its single control input by a causal, but otherwise arbitrary, coding and control scheme. We consider a single-input multiple-output (SIMO) channel between the encoder-controller and the decoder-controller which is lossless and imposes random time delay. We derive a lower bound on the minimum average feedback data rate that guarantees achieving a certain level of average quadratic performance over all possible realizations of the random delay. For the special case of a constant channel delay, we obtain an upper bound by proposing linear source-coding schemes that attain desired performance levels with rates that are at most 1.254 bits per sample greater than the lower bound. We give a numerical example demonstrating that bounds and operational rates are increasing functions of the constant delay. In other words, to achieve a specific performance level, greater channel delay necessitates spending higher data rate.              
\end{abstract}

\begin{IEEEkeywords}
Networked control systems, data rate constraints, time delay, information theory, optimal control.
\end{IEEEkeywords}

%
\IEEEpeerreviewmaketitle


\section{Introduction}
\label{intro}
Taking communication imperfections into account for analysis and design has proved to be an interesting topic within the area of control theory during recent years. This interest is motivated by advantages of communication networks over point-to-point wiring and, on the other side, by the complexity that communication constraints impose on classical control problems \cite{zhang2016survey}. Time delay, packet dropout and data rate constraints (quantization) are among prominent challenges \cite{nair2007feedback,zhang2013network,baillieul2007control,matveev2009estimation}.

Using an information-theoretic approach, \cite{martins2008feedback,martins2007fundamental} report primary derivations related to system performance. In these works, it is shown that the presence of a finite-capacity communication channel in a strictly causal feedback loop introduces a new performance limitation which differs from conventional Bode's formula by a constant quantifying channel information rate. Moreover, the authors derive inequalities among entropy rate of internal signals (inside the loop) and external signals
(outside the loop), resulting in a general performance bound which is affected by finite feedback capacity. Inspired by \cite{martins2008feedback,martins2007fundamental}, lower and upper bounds are derived on the minimum data rate that guarantees achieving a prescribed level of quadratic performance in \cite{silva2011framework,silva2011achievable,silva2016characterization}. These works consider noisy linear time-invariant (LTI) plants with Gaussian disturbances, controlled over an error-free digital channel without delay. In particular, \cite{silva2016characterization} shows that over all causal mappings which represent coding and control, the average data rate is bounded from below by the directed information rate generated by the mappings that render the sensor input and control output jointly Gaussian. Moreover, it is proved in \cite{silva2016characterization} that in an auxiliary LTI structure, the minimum signal-to-noise ratio (SNR) which guarantees stability and meeting a quadratic performance requirement gives the lower bound on the desired minimal data rate. For the upper bound analysis, \cite{silva2016characterization} suggests employing entropy-coded dithered quantizers (ECDQs). Such a simple coding scheme is designed based on the aforementioned SNR-constrained optimization giving the lower bound. Inspired by \cite{silva2011framework} and \cite{silva2016characterization}, the authors of \cite{tanaka2018lqg} present a method based upon semidefinite programming (SDP) to characterize the trade-off between directed information rate and linear quadratic Gaussian (LQG) performance in rate-constrained networked control systems (NCSs) with fully-observable multiple-input multiple-output (MIMO) plants. In \cite{stavrou2017upper}, the authors derive a lower bound on the zero-delay rate distortion function associated with vector-valued Gauss-Markov processes and mean-square error distortion constraint. Based on the separation principle, this bound is in fact the lower bound on the minimum data rate required for attaining LQG performance in control of fully observable plants. Then \cite{stavrou2017upper} utilizes the optimal realization that corresponds to the foreshadowed class of vector-valued Gaussian sources to derive an upper bound on zero-delay rate distortion function using variable-length entropy coding with lattice quantization. Similar ideas are employed in  \cite{kostina2016rate} for establishing bounds on minimum mutual informations, across a delay-free channel, that guarantee achieving specific linear quadratic regulator (LQR) performance levels. Specifically, \cite{kostina2016rate} derives the lower bound based on Shannon's lower bound and power entropy inequalities whereas the upper bound is established via variable-length coding and lattice-based quantization methods.

NCSs subject to network-induced delays are generally analyzed according to two methodologies: robustness and adaptation \cite{zhang2013network}. The aim in the robustness framework is deriving conditions for certain stability or performance requirements by constructing Lyapunov-Krosovskii functionals that do not incorporate time-stamp information as a variable. For instance, in \cite{du2014h}, stabilization and $H_{\infty}$ performance conditions for a singular cascade NCS are obtained. Fuzzy-model-based control is another approach in the robustness framework, where the rules are based on the size of delays, and the controller is required to be robust over the delay range \cite{khanesar2015adaptive,lu2015fuzzy}. In the adaptation framework, one method is modelling NCSs as stochastic switched systems. The recent results on stability and ${H_2}/{H_{\infty}}$ performance of Markov jump linear systems (MJLSs) are reported in \cite{xiong2014stability} and \cite{qiu2015network}, respectively. The second approach in this framework is predictive control; a method which is currently quite popular in NCSs. According to this technique, the actuator selects among a sequence of control commands based on the transmission delays experienced by them\cite{pang2014output,li2014network,yao2015wide}.

In all the aforementioned results on system performance, either the effect of channel delay is neglected, or the rate limitation is not taken into account. However, looking into the literature, one can find works investigating performance issues in NCSs with both rate constraints and network-induced delays (see, e.g., \cite{nakahira2016lq,han2016optimal,liu2016quantized,heemels2010networked}). Even so, a few has utilized the information-theoretic approach to treat systems with such limitations. For example, \cite{zhangrate} derives bounds on the minimum individual (non-asymptotic) rate needed to guarantee meeting an individual performance requirement (boundedness of the maximum $\ell_2$-norm of states).  

In this paper, we study the performance of a discrete-time LTI plant with Gaussian initial state in a loop with Gaussian exogenous inputs and random or constant channel delay on the feedback path. For the setup with random delay in the channel, we seek the infimum
average data rate required to achieve a prescribed qudratic performance
level. We show that the average data rate over all possible realizations of the delay is lower bounded by the average directed information rate. We prove for the random channel delay case that under certain stationarity assumptions, the average directed information rate can be stated in terms of average power spectral densities of the involved signals. We obtain a lower bound on the desired minimal average data rate which is stated as the average of a function of the power spectral densities of feedback path signals over all possile realizations of the delay. To establish all these results, we utilize the tools adopted in \cite{silva2011framework} and \cite{silva2016characterization}. However, compared to \cite{silva2016characterization} and \cite{silva2011framework}, the channel is not delay-free in our setup. In other words, we extend the information inequalities in \cite{silva2016characterization} to the case where there exists a random time delay between the sensor output and the control input.

For the setup with known constant delay in the channel, we show that the above lower bound on the infimum  average data rate required for attaining quadratic perfromance is equal to a function of infimum SNR of the channel over schemes comprised of LTI filters and AWGN channels with feedback and delay that meet the quadratic performance constraint. Our contribution in this case is showing how the presence of the channel delay affects the scheme yielding the lower bound. This gives an insight to the interplay between time delay, average data rate and performance in the considered NCS. We also prove that even over a channel with a constant delay, any admissible performance level can be achieved by an EDCQ-based linear coding scheme which generates an average data rate at most (approximately) $1.254$ bits per sample away from the corresponding lower bound. We illustrate via a numerical example that lower and upper bounds as well as empirical rates and entropies are all increasing functions of channel delay. This in turn implies that channels with larger delays demand higher average data rates to allow for attaining a certain system performance.

Compared to our previous works in \cite{barforooshan2017interplay} and \cite{barforooshan2017achievable}, first, we here study the case of random channel delay and  second, we employ a simpler proof than information inequalities and identities in \cite{barforooshan2017interplay} and \cite{barforooshan2017achievable}. In this work, we also show the effect of having a delay at different places in the loop on system signals. The last departure from our previous results is that we incorporate some eliminated proofs of \cite{barforooshan2017interplay} into this paper.

The remainder of the paper is organized as follows. Section~II introduces the notation. Section III formulates the main problem. Section~IV analyzes the lower bound problem for the setup with random channel delay. Section V derives a lower bound on the desired minimal data rate in the case of constant channel delay. The analysis of upper bound problem in the constant delay case is presented in Section~VI where the equivalence between systems with different delay locations is investigated. A numerical example is given in Section~VII. Finally, Section~IX concludes the paper.       


\section{Notation}
By $\mathbb{R}$, we denote the set of real numbers whose subset ${\mathbb{R}}^{+}$ represents the set of strictly positive real numbers. The set ${\mathbb{N}_{0}}$ is defined as ${\mathbb{N}_{0}}\triangleq\mathbb{N}\cup\{0\}$ where $\mathbb{N}$ symbolizes the set of natural numbers. The time index of every considered signal, denoted by $k$ in most cases, belongs to ${\mathbb{N}_{0}}$. Symbols ${\mathbf{E}}$, $\log$, $|.|$, and ${\norm{.}}_{2}$ represent operators for expectation, natural logarithm, magnitude and $H_2$-norm, respectively. Moreover, ${\lambda}_{\min}(S)$ and ${\lambda}_{\max}(S)$ are respectively the largest and smallest eigenvalues of the square matrix $S$ for which the element on the $i$-th row and $j$-th column is denoted by $[S]_{i,j}$. In addition, ${\beta}^{k}$ is shorthand for ${\beta}(0),\dots,{\beta}(k)$ where ${\beta}(k)$ denotes the $k$-th sample of a discrete-time signal. Furthermore, for the time-dependent set ${\alpha}(i),i\in{{\mathbb{N}_0}}$, ${\alpha}^{k}$ is defined as ${\alpha}^{k}\triangleq{{\alpha}(0)\times\dots\times{\alpha}(k)}$. However, if $\alpha$ is a fixed set, then ${\alpha}^{k}\triangleq{{\alpha}\times\dots\times{\alpha}}$ ($k$ times).

Random variables and processes are vector valued, unless otherwise stated.
Take $v$ and $q$ into account as two random variables with known marginal and joint probability distribution functions (PDFs). Their joint PDF is represented by $f(v,q)$ while the marginal PDFs of $v$ and $q$ are symbolized by $f(v)$ and $f(q)$, respectively. The conditional PDf of $v$ given $q$ is denoted by $f(v|q)$ and ${\mathbf{E}}_{v}(.)$ is the operator for the expectation with respect to the distribution of $v$. We define the differential entropy of $v$ and the conditional differential entropy of $v$ given $q$ as $h(v)\triangleq{-{\mathbf{E}}_{v}}(\log{f(v)})$ and $h(v|q)\triangleq{-{\mathbf{E}}_{v,q}}(\log{f(v|q)})$, respectively. The  mutual information between $v$ and $q$ is symbolized by by $I(v;q)$ and 

\begin{figure}[thpb]
	\centering
	\includegraphics[width=8cm]{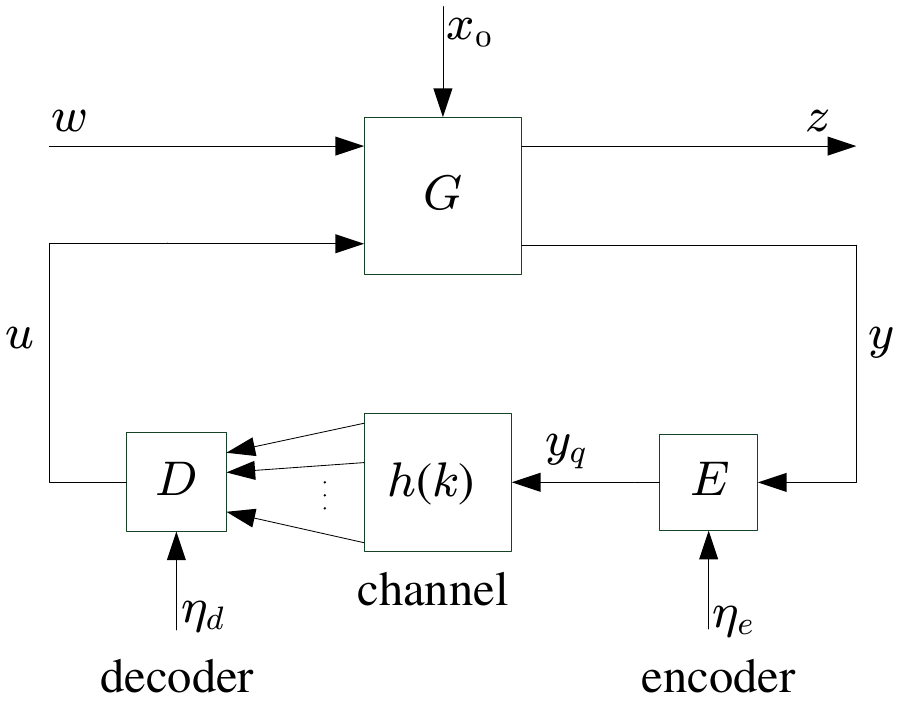}
	\caption{Considered NCS with a channel imposing random delay}
	\label{fig21}
\end{figure}
defined as $I(v;q)\triangleq{-{\mathbf{E}}_{v,q}}(\log({f(v)f(q)/f(v,q)}))$. Moreover, the definition of the conditional mutual information between random variables $v$ and $q$ given the random variable $r$ is given by $I(v;q|z)\triangleq{I(v,r;q)-I(r;q)}$. All the information-theoretic definitions presented in this paragraph are standard and follow \cite{cover2012elements}.

We call the random process $\xi$ asymptotically wide-sense stationary (AWSS) if $\lim_{k\to\infty}{{\mathbf{E}}[\xi(k)]}={\nu}_{\xi}$ and
${C_\xi}\triangleq{{R}_{\xi}(0)}$ and $\lim_{k\to\infty}{{\mathbf{E}}[(\xi(k+\tau)-{\mathbf{E}}[\xi(k+\tau)]){(\xi(k)-{\mathbf{E}}[\xi(k)])}^{T}]}={R}_{\xi}(\tau)$ hold, where ${\nu}_{\xi}$ is a finite constant. Accordingly, the steady-state covariance matrix and the steady-state variance of $\xi$ are defined as $\sigma_\xi^2\triangleq\mathrm{trace}({C_\xi})$, respectively. For the scalar random sequence ${x}_{1}^{k}\triangleq{[x(1)\dots x(k)]^T}$, we define the covariance matrix as $C_{{x}_{1}^{k}}={{\mathbf{E}}[({{x}_{1}^{k}}-{\mathbf{E}}[{{x}_{1}^{k}}]){({{x}_{1}^{k}}-{\mathbf{E}}[{{x}_{1}^{k}}])}^{T}]}$. Assume that ${P_n},{Q_n}\in{{\mathbb{R}}^{n\times n}}$ are square matrices. Then the sequences $\{P_n\}_{n=1}^{\infty}$ and $\{Q_n\}_{n=1}^{\infty}$ are called  asymptotically equivalent if and only if they satisfy the following expression for finite $\varrho$: 
\begin{IEEEeqnarray*}{rl}
	\begin{split}
		&\lim_{n\to\infty}\frac{1}{n}\sum_{i=1}^{n}\sum_{j=1}^{n}{{|{[{P_n}-{Q_n}]}_{i,j}|}^{2}}=0\\
		&|{\lambda}_{\max}(P_n)|,|{\lambda}_{\max}(Q_n)|\leq{\varrho},\quad\forall{n\in{\mathbb{N}}}. 
	\end{split}
\end{IEEEeqnarray*}
\section{Problem Formulation}\label{PF}
We consider the feedback loop of Fig.~{\ref{fig21}} where the plant is LTI with one control input and one sensor output denoted by $u\in\mathbb{R}$ and $y\in\mathbb{R}$, respectively. The plant $G$ is disturbed by a vector-valued zero-mean white noise which is represented by $w\in \mathbb{R}^{n_{w}}$ and has identity covariance matrix, i.e. ${C_w}=I$. Moreover, as depicted in Fig.~\ref{fig21}, the plant outputs the vector-valued signal $z\in \mathbb{R}^{n_{z}}$ upon which the performance measure is characterized.
The relationship between the mentioned set of inputs and outputs is described by a transfer-function matrix as follows:

\begin{equation}\label{eq1}
\left[
\begin{array}{c}
	z\\
	y
\end{array}\right]=\left[
\begin{array}{lr}
	G_{11}&G_{12}\\
	G_{21}&G_{22}
\end{array}\right]\left[
\begin{array}{c}
	w\\
	u
\end{array}\right],
\end{equation}
where the dimensionality of each $G_{ij}$ is determined by the dimensions of corresponding pair of inputs and outputs. So ${n_z}\times{n_w}$, ${n_z}\times{1}$, ${1}\times{n_w}$ and ${1}\times{1}$ are the dimensions for $G_{11}$, $G_{12}$, $G_{21}$ and $G_{22}$, respectively.
\begin{assumption}\label{ass21}
Every entry of the transfer-function matrix in (\ref{eq1}) is proper with no unstable hidden modes. Moreover, ${G_{22}}$, which describes the single-input single-output (SISO) open-loop system from $u$ to $y$, is strictly proper. The initial states of the plant denoted by $x_0=\{x(-{h_{\max}}),\dots,x(0)\}$ are jointly Gaussian with and independent of the disturbance signal $w$ and has a finite differential entropy.    
\end{assumption} 
As depicted in Fig.~{\ref{fig21}}, the output of the plant, $y$, is processed into a binary word by the encoder $E$ and transmitted over the error-free channel. Such transmission is accompanied with a random delay. Let $h(k)$ denote the delay experienced by the binary word ${y_q}(k)$ constructed at time $k$ at the encoder. We assume that $h(k)$ is an independent and identically distributed (i.i.d.) process which has a bounded support at each time step, i.e., $h(k)\in\{{h_1},\dots,{h_m}\}$, $\forall{k}\in{\mathbb{N}_0}$ where ${h_i}<{h_{i+1}}$ $(i=1,2,\dots,m-1)$. In order to avoid unnecessary notational complexity and without loss of generality, we set ${h_1}$ as ${h_1}=0$ and ${h_m}$ as ${h_m}={h_{\max}}$. The marginal distribution of the delay is assumed to be known and described by $Pr\{h(k)=h_j\}=\alpha_j$ where $\sum_{j=1}^{m}{\alpha_j}=1$. Such characteristics introduce a channel with the following input-output relationship:
\begin{equation}\label{eq7}
{u_q}(k)=[{{y_q}(i)]}_{i\in{S(k)}},\qquad k\in{\mathbb{N}_0}
\end{equation} 
where $S(k)$ is defined as 
\begin{equation}\label{eqnew1}
S(k)\triangleq{\{i:i+h(i)=k\}}
\end{equation} 
for every $k,i\in{\mathbb{N}_0}$. Denoting the cardinality of $S(k)$ by $s(k)$, we can imply form (\ref{eq7}) that $u_q(k)$ is a vector comprised of $s(k)\leq h_{\max} + 1$ binary words which specifies the output of the channel at time $k$. Note that $s(k)$ is a random variable depending on the channel delay. We assume that ${y_q}(i)$ is discarded at the decoder-controller side if $i<0$. Moreover, under aforementioned circumstances, binary words transmitted over the considered channel are not necessarily received in the same order they were emitted. 
\begin{figure}[h] 
	\centering    
	\includegraphics[width=8cm]{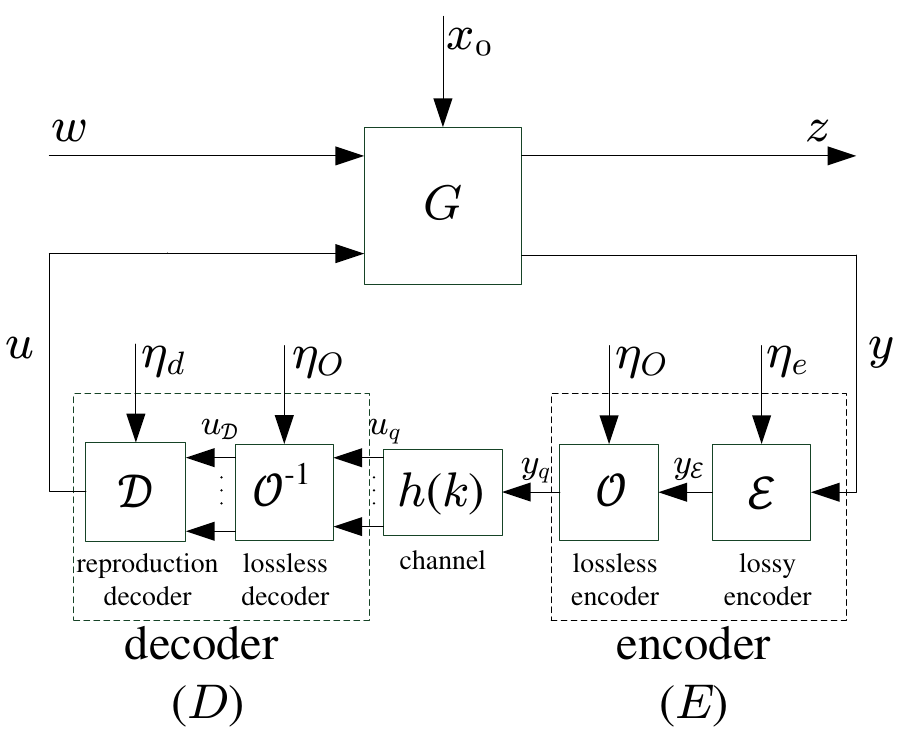}
	\caption{Considered NCS with a detailed model of coding and control scheme}
	\label{fig3}
\end{figure}
It should be also emphasized that the channel does not allow for any data loss. The average data rate across the channel is defined as
\begin{equation}\label{eq4}
\mathcal{R}
\triangleq\lim_{k \to \infty}
\frac{1}{k}
\sum_{i=0}^{k-1}R(i),
\end{equation}
where $R(i)$ indicates the expected length of the binary word ${y_q}(i)$.

A more detailed presentation of the feedback path in the NCS of Fig.~\ref{fig21} is provided by Fig.~\ref{fig3}. As depicted, the encoder-controller is comprised of a lossy and a lossless component. The lossy part $\mathcal{E}$ outputs the symbol $y_{\mathcal{E}}$ according to the following dynamics:
\begin{equation}\label{eq24}
{{y_{\mathcal{E}}}(k)}={\mathcal{E}_k}({y^k},{\eta_{e}^{k}}),
\end{equation}
where ${{\eta_{e}}(k)}\in{{{\Lambda}_e}(k)}$ symbolizes the side information at time $k$ at the lossy encoder. ${\mathcal{E}}_k:{\mathbb{R}}^{k+1}\times{{{{{\Lambda}_{e}^{k}}}}}\mapsto{{\mathcal{A}}_{s}}$ is a deterministic map and ${{\mathcal{A}}_{s}}$ represents a fixed countable set. At each time instant, the encoder is assumed to know the time delays experienced by previous binary words and the time delay of the current binary word to be sent over the channel. Therefore, ${h}^{k}$ is known at the encoder $\forall{k}\in{{\mathbb{N}}_0}$. This implies that ${y_{\mathcal{E}}(k)}$  can be reconstructed perfectly $h(k)$ steps later at the decoder if ${y_q}(k)$ is constructed by using ${y_{\mathcal{E}}(k)}$ and only those samples of ${y_{\mathcal{E}}^{k-1}}$ which will be already available at the decoder at time $k+h(k)$. Note that having access to ${y_q^k}$ at the decoder at the time $k+h(k)$ is not assured. So the lossless encoder $\mathcal{O}$ outputs the binary symbol ${y_q}$ based on 
\begin{equation}\label{eq25}
{{y_q}(k)}={\mathcal{H}_{\mathcal{E}k}}({y_{{\mathcal{E}}}}(k),{y_{{\mathcal{E}}f}(k)},{\eta_{o}^{k}}),
\end{equation}
in which ${y_{{\mathcal{E}}f}(k)}$ is a sequence comprising the elements of ${y_{\mathcal{E}}^{k-1}}$ for which the associated binary words will have reached the decoder by the time $k+h(k)$, i.e., $\{{y_{\mathcal{E}}}(i):i\in{{\mathbb{N}}_0},i\leq{k-1}, i+h(i)\leq{k+h(k)}\}$. Moreover, ${\eta_{o}(k)}\in{\Lambda_{o}(k)}$, and ${\mathcal{H}_{\mathcal{E}k}}:{\mathcal{A}}_{s}^{f(k)}\times{{{{\Lambda}_{o}^{k}}}}\mapsto{\mathcal{A}}(k)$ is an arbitrary deterministic mapping where ${k-{h_{\max}}+h(k)+1}\leq{f(k)}\leq{k+1}$. So $f(k)-1$ specifies the cardinality of the sequence ${y_{{\mathcal{E}}f}(k)}$. Note that since no dropout occurs during data transmission, $y_q^{k-{h_{\max}}+h(k)}$ will certainly have been received at the decoder by the time $k+h(k)$. In addition, ${\mathcal{A}}(k)$ is a countable set of prefix-free binary code words, which specifies the input alphabet of the channel at each time instant.

On the receiver side, ${u_{q}}(k)$ is available as the input to the lossless decoder. This decoder, shown by $\mathcal{O}^{-1}$, generates ${u_{\mathcal{D}}}$ as
\begin{equation}\label{eq26}
{{u_{\mathcal{D}}}(k)}={\mathcal{H}_{\mathcal{D}k}}({u_{qf}}(k),{\eta_{o}^{k}},S(k)), 
\end{equation}
where ${u_{qf}}(k)$ is a sequence comprised of elements of $u_q^k$ that have time indices less than or equal to the largest time index of $y_q$ in ${u_q}(k)$, i.e., ${u_{qf}}(k)\triangleq\{{y_q}(i):i\in{{\mathbb{N}}_0},{y_q}(i)\in{u_q^k},i\leq{m(k)}\}$ where ${m(k)}={{{\max}S(k)}}$, $\forall{k}\in{\mathbb{N}_0}$. Such selection of data for lossless decoding is in accordance with the information utilized in (\ref{eq25}) for encoding. Furthermore, ${\mathcal{H}_{\mathcal{D}k}}:{{\mathcal{A}}^{g(k)}\times{{{{\Lambda}_{o}^{k}}}}\times{{\mathbb{N}_0}^{s(k)}}}\mapsto{\mathcal{A}_{s}^{s(k)}}$, where ${k-{h_{\max}}+1}\leq{g(k)}\leq{k+1}$, represents an arbitrary deterministic mapping. It should be noted that according to the channel model, $g(k)$ is a random variable denoting the cardinality of ${u_{qf}}(k)$ and  $\sum_{i=0}^{k}s(i)\leq{k}$ holds for all $k\in{\mathbb{N}_0}$. Moreover, based on the definition of $u_{qf}$ and (\ref{eq25}), the information provided by $({u_{qf}}(k),{\eta_{o}^{k}},{S(k)})$ is enough for the lossless decoder to reconstruct every element of $\{y_{\mathcal{E}}(i)\}_{i\in{S(k)}}$ perfectly. Therefore 
\begin{equation}
	{{u_{\mathcal{D}}}(k)}=[y_{\mathcal{E}}(i)]_{i\in{S(k)}},\qquad k\in{\mathbb{N}_0}
\end{equation}
where $S(k)$ is defined as in (\ref{eqnew1}). Indeed for such reconstruction, the knowledge of the delay is required at the decoder. Hence, we further assume that the decoder is provided by $S(k)$ through for example timestamping. Finally, the decoder-controller gives the control input via
\begin{equation}\label{eq271} 
{u(k)}={\mathcal{D}_k}({u_{\mathcal{D}}^{k}},{\eta_{d}^{k}}).
\end{equation}
where ${\eta_{d}(k)}$ signifies the side information available at the decoder at time $k$ and is contained in the well-defined set ${{{\Lambda}_d}(k)}$. So ${\eta_d}(k)$ satisfies ${{\eta_d}(k)}\in {{{\Lambda}_d}(k)}$. Moreover, ${\mathcal{D}_k}:{{\mathcal{A}}_{s}^{{t_u}(k)}}$ $\times{{{{\Lambda}_{d}^{k}}}}\mapsto{\mathbb{R}}$, ${k-{h_{\max}}+1}\leq{{t_u}(k)}\leq{k+1}$, is an arbitrary deterministic mapping where ${t_u}(k)$ is the cardinality of $S^k$. It should be noted that ${\Lambda_{o}(k)}$ in ${\eta_{o}(k)}\in{\Lambda_{o}(k)}$ is defined as ${\Lambda_{o}(k)}\triangleq{{{\Lambda}_e}(k)}\cap{{{\Lambda}_d}(k)}$. We state some additional properties of the setting described above in the following remarks.
\begin{remark}\label{rem2}
It can be implied from (\ref{eq24})-(\ref{eq271}) that $u(k)$ and $u^k$ are functions of $({y^{l_k}},{\eta_e}^{l_k},{\eta_{d}^{k}})$ where ${l_k}=\max{S^k}$ for every $k\in{\mathbb{N}_0}$. It thus follows from the definition of ${S(k)}$ and assuming no dropout in the considered channel that ${k-{h_{\max}}}\leq{l_k}\leq{k}$. This implies that the controller has access to the largest and smallest amount of sensor information when the channel delay is zero and $h_{\max}$, respectively.
\end{remark}
   
\begin{remark}
It can be implied from the definition of ${S}(k)$ in (\ref{eqnew1}) that ${u_q}(k)$ can have at most ${h_{\max}}+1$ entries at each time step. So the number of the words that can be received at the decoder at each time instant belongs to the set $\{0,\dots,{h_{\max}}+1\}$. Therefore, since the channel input $y_q$ is a scalar process, (\ref{eq7}) describes a single-input multiple-output (SIMO) channel. Moreover, $S(k)$, as a stochastic process, cannot be i.i.d because in the considered channel, no transmitted binary word is received at the decoder more than once.  
\end{remark}
\noindent
For further analysis, we consider the following assumption. 
\begin{assumption}\label{ass23}
	At each time instant $k\in{\mathbb{N}_0}$, the side information pair $({{\eta}_{e}}(k),{{\eta}_{d}}(k))$ together with $h(k)$, and consequently $S(k)$, are statistically independent of $(x_0,w(k))$. Therefore, it can be implied from the dynamics of the system that $I(u(k);{y(k-{h_i})}\mid{u^{k-1}})=0$ for any ${h_i}\in\{1,...,{h_{\max}}\}$ with $k-{h_i}<0$. Moreover, upon knowledge of $u^{i}$, $\eta_{d}^{i}$ and ${S^i}$, the decoder is invertible. It means that for each ${i}\in{\mathbb{N}_0}$, there exists a deterministic mapping $Q_{i}$ such that $u_{q}^{i}=Q_{i}(u^i,\eta_{d}^{i},{S^i})$.     
\end{assumption} 
\begin{remark}\label{rem1}
  In Appendix~\ref{appb}, we will prove that for the architecture of Fig.~\ref{fig3}, any encoder and non-invertible decoder with mappings ${\mathcal{E}}$, $\mathcal{O}$, $\mathcal{O}^{-1}$ and ${\mathcal{D}}$, can be replaced by another set of mappings with the same input-output relationship and lower average data rate where the decoder is invertible .
\end{remark}
For the purpose of expressing the information rate in terms of spectral densities of the signals of the system, we use the following notion of stability:
\begin{definition}\label{def1}
 A scalar AWSS process $x$ is called strongly asymptotically wide-sense stationary (SAWSS) if its covariance matrix is asymptotically equivalent to the covariance matrix of the wide sense stationary (WSS) process, say $\bar{x}$, to which it converges, i.e., ${\{C}_{x_1^n}\}_{n=1}^{\infty}$ and ${\{C}_{{\bar{x}}_1^n}\}_{n=1}^{\infty}$ are asymptotically equivalent. Furthermore, in an SAWSS NCS, all internal signals are SAWSS and their cross-covariance matrices are asymptotically equivalent to the cross-covariance matrices of corresponding WSS processes to be converged to.
\end{definition}   
\noindent
Clearly, SAWSS-ness implies AWSS-ness but not vice versa; for both signals and systems. For each coding scheme satisfying (\ref{eq24})-(\ref{eq271}) and rendering the NCS of Fig.~\ref{fig21} SAWSS, the steady-state variance of the output $z$ is a random variable which depends on the realization of $h(k)$. The same goes for the average data rate. We make explicit such dependence by writing $\sigma_z^2(h^k)$ and $\mathcal{R}(h^k)$, and consider the means of these variables (over all realizations of $h^k$) as our performance measure and data rate of interest, respectively. Such notions of performance and rate, represented by ${{\sigma}_{za}^{2}}$ and $\mathcal{R}_a$ respectively, are formulated as follows:
\begin{align}\label{511}
\begin{split}
{{\sigma}_{za}^{2}}&=\Sigma_{\substack{{h^k}\in{\mathfrak{H}}\\
		k\to\infty}}{Pr(h^k){{\sigma}_{z}^{2}}(h^k)}\\
{{\mathcal{R}}_{a}}&=\Sigma_{\substack{{h^k}\in{\mathfrak{H}}\\
		k\to\infty}}{Pr(h^k){\mathcal{R}}(h^k)}\\			             
\end{split}
\end{align}
where $\mathfrak{H}$ denotes the support set for possible realizations of the delay $h(k)$. Moreover, ${\mathcal{R}}(h^k)$ and ${{\sigma}_{z}^{2}}(h^k)$ indicate that the average data rate and steady-state variance are functions of delay. Generally speaking, we are interested in finding the minimal ${{\mathcal{R}}_{a}}$ for which having a bounded $\sigma_{za}^2$ is feasible. Let ${D}_{\inf}(h_r)$ denote the smallest average steady-state  variance of $z$ that can be achieved, when the random delay $h(k)$, with the aforementioned properties, is present in the channel. Hence, ${D}_{\inf}(h_r)$ is obtained by minimizing the average steady-state variance of $z$ over all (possibly nonlinear and time-varying) settings $u(k)={\mathcal{K}}_{k}( y^{l_{k}})$ that render the NCS of Fig.~\ref{fig21} SAWSS. Note that ${l_k}$ is defined as in Remark~\ref{rem2}. Under the condition that Assumption~\ref{ass21} holds, the problem of our interest is to find
\begin{equation}\label{eq51}
{\mathcal{R}_a}(D)=\inf_{{{\sigma}_{za}^{2}}\leq{D}}
\mathcal{R}_a,
\end{equation}
where $D\in{(D_{\inf}(h_r),\infty)}$, and ${{\sigma}_{za}^{2}}$ represents the average staedy-state variance of the output $z$ over all realizations of the delay. The feasible set of the optimization problem in (\ref{eq51}) is comprised of all encoder-controller  and decoder-controller pairs described by (\ref{eq24})-(\ref{eq271}), satisfying Assumption~\ref{ass23} and rendering the NCS of Fig.~\ref{fig21} SAWSS. 
\begin{remark}\label{rem7}
	It is straightforward to see from (\ref{eq24})-(\ref{eq271}) that the concatenation of the decoder-controller pair and the channel in the NCS of Fig.~\ref{fig21} is equivalent to a decoder-controller pair with the same mapping and side information that applies a time delay with same properties as characterized in (\ref{eq7}), on its received data, and that is followed by a delay-free channel. So the system depicted in Fig.~\ref{fig21} is equivalent to the feedback loop of Fig.~\ref{fig2aux} in which the encoder and the plant are the same and the inputs have the same properties as in Fig.~\ref{fig21}.
\end{remark}
\begin{figure}
	\centering
	\includegraphics[width=8cm]{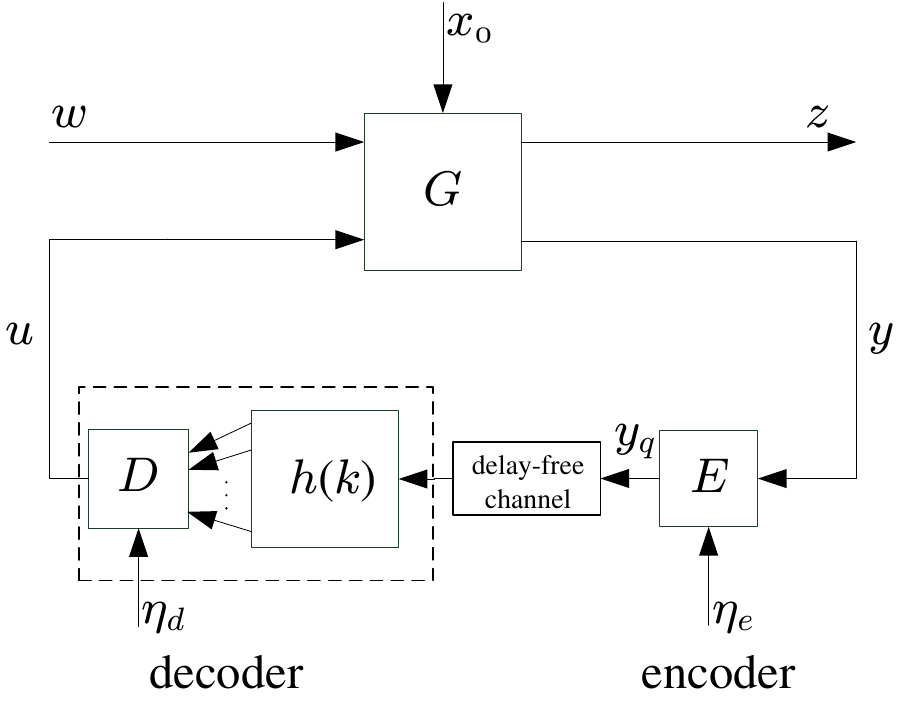}
	\caption{Auxiliary system equivalent to the main NCS in the random delay case}
	\label{fig2aux}
\end{figure} 
The equivalence pointed out in Remark~\ref{rem7} between systems of Fig.~\ref{fig21} and the NCS of Fig.~\ref{fig2aux} will assist us deriving a lower bound on the average data rate ${\mathcal{R}_a}$ in the next Section.
 \section{lower bound problem in the presence of random delay}
 In this section, we establish a lower bound on ${{\mathcal{R}}_{a}(D)}$. To do so, we derive inequalities and identities that describe the relationship between the flow of information and system performance in the NCS of Fig.~\ref{fig21}. Therefore, we will update fundamental derivations in \cite{silva2016characterization,barforooshan2017achievable} for the case where the channel delay is randomly distributed. As the first result, we show how the average data rate ${\mathcal{R}_a}$ is bounded from below in the following theorem.
  \begin{theorem}\label{th5}
 	Consider the feedback loop depicted in Fig.~\ref{fig21} for which Assumptions~\ref{ass21} and \ref{ass23} hold. Then 
 	\begin{equation}\label{eq291}
 	\begin{split}
 	\mathcal{R}_a
 	\geq{{I_{\infty}^{h_a}}(y\to{u})}=\lim_{k \to \infty}
 	\frac{1}{k}
 	\Sigma_{{h^{k-1}}\in{\mathfrak{H}}}[{Pr(h^{k-1})}\Sigma_{i=0}^{k-1}{I({{u}(i);y^{l_{i}}}\mid{{u^{i-1}}}})]
 	\end{split}
 	\end{equation}
 	where $I(.;.\mid{.})$ represents conditional mutual information. According to \cite{derpich2013fundamental}, ${{I_{\infty}^{h_a}}(y\to{u})}$ specifies the average directed information rate across the forward channel from $y$ to $u$ in the NCS of Fig.~\ref{fig21} over all possible realizations of the channel delay.   
 \end{theorem}
\begin{proof}
	It can be implied from \cite[Theorem~3.1]{silva2016characterization} that, for each realization $h^\infty$ of the delays, the average data rate (\ref{eq4}) in the feedback loop of Fig.~\ref{fig2aux} is bounded from below as
	\begin{equation}\label{eq292}
	\mathcal{R}(h^\infty)
	\geq{{I_{\infty}}(y\to{u})}=\lim_{k \to \infty}
	\tfrac{1}{k}
	\Sigma_{i=0}^{k-1}I(u(i);{y^{i}}\mid{u^{i-1}}).
	\end{equation}
	Based upon the chain rule of mutual information, the bound in (\ref{eq292}) can be restated as 
	\begin{align}\label{eq293}
	\begin{split}
	\mathcal{R}(h^\infty)\geq
	\lim_{k \to \infty}
	\tfrac{1}{k}
	\Sigma_{i=0}^{k-1}[{I(u(i);{y^{l_i}}\mid{u^{i-1}})}+I(u(i);{y_{{l_i}+1}^{i}}\mid{u^{i-1},{y^{l_i}}})],
	\end{split}
	\end{align}
	where the definition of $l_i$ is given in Remark~\ref{rem2}. From the dynamics of the plant, we can easily conclude that the sequence ${y_{{l_i}+1}^{i}}$ is only a function of $x(0)$ and $w$, once $u^{i-1}$ is given. Furthermore, it stems from (\ref{eq7})-(\ref{eq271}) that upon the knowledge of $y^{l_i}$, side informations $\eta_{d}^{i}$ and $\eta_{e}^{l_i}$ will be the only variables describing $u(i)$, $\forall{i}\in{\mathbb{N}_0}$. Latter observations together with the fact that Assumption~\ref{ass23} holds for the system of Fig.~\ref{fig21} yield the conclusion that the rightmost term of (\ref{eq293}) amounts to zero, $\forall{i}\in{\mathbb{N}_0}$. So we have
	\begin{align}\label{eq294}
	\begin{split}
	\mathcal{R}(h^\infty)\geq
	\lim_{k \to \infty}
	\tfrac{1}{k}
	\Sigma_{i=0}^{k-1}&{I(u(i);{y^{l_i}}\mid{u^{i-1}})}
	\end{split}
	\end{align}
	for the NCS of Fig.~\ref{fig2aux}. Now by averaging both sides of (\ref{eq294}) with respect to the delay realizations, as in (\ref{511}), and noting that the feedback loop of Fig.~\ref{fig2aux} is equivalent to the system of Fig.~\ref{fig21}, based on Remark~\ref{rem7}, our claim follows immediately.    
\end{proof}
 The next lemma shows that joint Gaussianity of two signals lowers the directed information rate between them when these are connected through a channel with random delay.
  \begin{lemma}\label{lemma6}
 	Suppose that the NCS of Fig.~\ref{fig21} satisfies Assumption~\ref{ass21} and Assumption~\ref{ass23}. For this system, if $(x(0),w,u,y)$ represents a jointly second-order set of processes, then the following holds:
 	\begin{equation}\label{eq295}
 	{I_{\infty}^{h_a}}(y\to{u})\geq{I_{\infty}^{h_a}}({y_G}\to{u_G}),
 	\end{equation}
 	where $y_G$ and $u_G$ symbolize the Gaussian counterparts of $y$ and $u$, respectively, in a way that $(x(0),w,u_G,y_G)$ are jointly Gaussian with the same first-and second-order (cross-) moments as $(x(0),w,u,y)$. 	
 \end{lemma} 
 \begin{proof}
	According to \cite[Lemma~3.1]{silva2016characterization}, the directed information rate from sensor output to the control input in the auxiliary NCS of Fig.~\ref{fig2aux} is bounded as follows:
	\begin{equation}\label{eq296}
	{{I_{\infty}}(y\to{u})\geq{I_{\infty}}({y_G}\to{u_G})},
	\end{equation}
	where ${I_{\infty}}(y\to{u})$ and ${I_{\infty}}({y_G}\to{u_G})$ are defined as in (\ref{eq292}). We conclude based on (\ref{eq7})-(\ref{eq271}), the dynamics of the plant and the system of Fig.~\ref{fig2aux} satisfying Assumption~\ref{ass23} that $I(u(i);{y_{{l_i}+1}^{i}}\mid{u^{i-1},{y^{l_i}}})=0, \forall{i}\in{\mathbb{N}_0}$. This together with the chain rule of mutual information lead to 
	\begin{equation}\label{eq297}
	\begin{split}
	\lim_{k \to \infty}
	\tfrac{1}{k}
	\Sigma_{i=0}^{k-1}{I(u(i);{y^{l_i}}\mid{u^{i-1}})}\geq\lim_{k \to \infty}
	\tfrac{1}{k}
	\Sigma_{i=0}^{k-1}{I({u_G}(i);{{y_G}^{l_i}}\mid{{u_G}^{i-1}})}.
	\end{split} 
	\end{equation}
	Now the proof is complete by taking average over all possible realizations of the delay from both sides of (\ref{eq297}) and considering that based on Remark~\ref{rem7}, the system of Fig.~\ref{fig21} is equivalent to the feedback loop in Fig.~\ref{fig2aux}. 	
\end{proof}
 If the above Gaussian signals are stationary as well, then the average directed information rate can be stated in terms of the average power spectral density of the involved signals. The next lemma will state such result formally.
 \begin{lemma}\label{lemma8}
 	Suppose that the control input $u$ in the NCS of Fig.~\ref{fig21} is SAWSS for every realization of the channel delay. For each realization, assume that there exists a $\mu>0$ in such a way that $|{{{\lambda}_{\min}}(C_{{u}_{1}^{n}})}|\geq{\mu}$, $\forall{n}\in{\mathbb{N}}$. Let further consider the sensor output $y$ jointly AWSS with $u$. Then the average directed information rate is equal to an integral term as follows:
 	\begin{equation}\label{eq298}
 	\begin{split}
 	{{I_{\infty}^{h_a}}(y\to{u})}=\sum_{\substack{{h^k}\in{\mathfrak{H}}\\
 			k\to\infty}}{Pr(h^k)}[\frac{1}{4\pi}\int_{-\pi}^{\pi} \log\big(\frac{{S_{\check{u}}}(e^{j\omega},h^k)}{{\sigma}_{\psi}^{2}(h^k)}\big)d\omega],
 	\end{split}
 	\end{equation}
 	where $\psi$ is a Gaussian AWSS process that has independent samples. Such a random process is described as 
 	\begin{equation}\label{eq299}
 	\psi(k)\triangleq{u(k)-\tilde{u}(k)},
 	\tilde{u}(k)\triangleq{{\mathbf{E}}[u(k)\mid{{y^{l_k}},{u^{k-1}}}]}
 	\end{equation}
 	for each realization of the random delay. Moreover, ${S_{\check{u}}}$ represents the steady-state power spectral density of $u$.
 \end{lemma}
 \begin{proof}
	It can be deduced from \cite[Lemma~3.2]{silva2016characterization} that in the NCS of Fig.~\ref{fig2aux}, the following holds for the directed information rate :
	\begin{equation}\label{eq300}
	{{I_{\infty}}(y\to{u})}=\frac{1}{4\pi}\int_{-\pi}^{\pi} \log\big(\frac{{S_{\bar{u}}}(e^{j\omega})}{{\sigma}_{n}^{2}}\big)d\omega,
	\end{equation}
	in which $n$ is a Gaussian AWSS process with independent samples and ${{I_{\infty}}(y\to{u})}$ is defined as in (\ref{eq292}). The noise $n$ is calculated as follows:
	\begin{equation}\label{eq301}
	n(k)\triangleq{u(k)-\hat{u}(k)},
	\hat{u}(k)\triangleq{{\mathbf{E}}[u(k)\mid{{y^{k}},{u^{k-1}}}]}.
	\end{equation} 
	As already mentioned before, based on the plant dynamics, the knowledge of $u^{i-1}$ will render ${y_{{l_i}+1}^{i}}$ dependent only on $x(0)$ and $w^i$ for any $i\in{\mathbb{N}_0}$. Moreover, according to (\ref{eq7})-(\ref{eq271}), knowing $y^{l_i}$, one can determine $u(i)$ by only figuring out  ${\eta_d}^{i}$ and ${\eta_e}^{l_i}$, $\forall{i}\in{\mathbb{N}_0}$. Since, based on Assumption~\ref{ass23}, $(x_0,w)$ and $({\eta_d},{\eta_e})$ are independent, the following is yielded:
	\begin{equation}\label{eq302}
	\begin{split}
	\lim_{k \to \infty}
	\tfrac{1}{k}
	\Sigma_{i=0}^{k-1}{I(u(i);{y^{l_i}}\mid{u^{i-1}})}=\frac{1}{4\pi}\int_{-\pi}^{\pi} \log\big(\frac{{S_{\bar{u}}}(e^{j\omega})}{{\sigma}_{n}^{2}}\big)d\omega,\\
	{{\mathbf{E}}[u(k)\mid{{y^{k}},{u^{k-1}}}]}={{\mathbf{E}}[u(k)\mid{{y^{l_k}},{u^{k-1}}}]}.
	\end{split} 
	\end{equation}
	From (\ref{eq302}), it can be concluded that $n(k)$ is actually equal to $\psi(k)$ as in (\ref{eq299}) for the NCS of Fig.~\ref{fig2aux}. Now,  our claim is given by taking average from both sides of upper (\ref{eq302}) and noting that based on Remark~\ref{rem7}, the systems in Fig.~\ref{fig2aux} and Fig.~\ref{fig21} are equivalent. 
\end{proof}
We are now ready to present a lower bound on ${\mathcal{R}_a}(D)$. A corollary follows: 
\begin{corollary}
	Suppose that the NCS of Fig.~\ref{fig21} satisfies Assumption~\ref{ass21}. Then ${\mathcal{R}_a}(D)$ is lower bounded as follows:
	 	\begin{equation}\label{eq3}
	\begin{split}
	{\mathcal{R}_a}(D)\geq{\inf_{{{\sigma}_{za}^{2}}\leq{D}}}\sum_{\substack{{h^k}\in{\mathfrak{H}}\\
			k\to\infty}}{Pr(h^k)}\frac{1}{4\pi}\int_{-\pi}^{\pi} \log\big(\frac{{S_{\check{u}}}(e^{j\omega},h^k)}{{\sigma}_{\psi}^{2}(h^k)}\big)d\omega,
	\end{split}
	\end{equation}
where $\psi$ and $\check{u}$ are defined as in (\ref{eq299}) and the infimum is restricted to all mappings staisfying (\ref{eq24})-(\ref{eq271}) and Assumption~\ref{ass23}, and producing signals $y$ and $u$ with propoerties as stated in Lemma~\ref{lemma8}.
\end{corollary}
\begin{proof}
	The claim follows immediately from Theorem~\ref{th5} and Lemma~\ref{lemma8}.
\end{proof} 
\section{Lower bound problem in the case of the constant delay}
 In this section, we consider the same NCS as described in Section~III but with a channel that imposes a known constant delay, say $h$ steps, on the transmitted data. The corresponding feedback loop is depicted by Fig.~\ref{fig1}. 
\begin{figure}
\centering
\includegraphics[width=8cm]{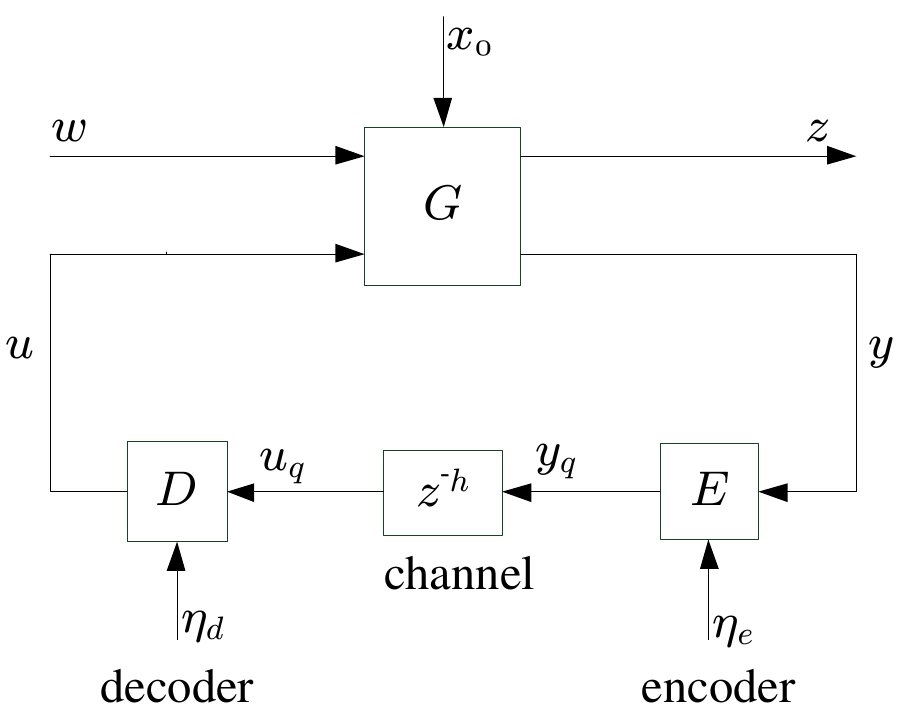}
\caption{Considered NCS in the constant delay case}
\label{fig1}
\end{figure}
 The problem we investigate here is a special case of the problem formalized in (\ref{eq51}) where the channel delay is constant and therefore, there is only one realization for the channel delay. In this case, we consider the notation ${{\mathcal{R}}_{a}}(D)={\mathcal{R}}(D)$ and ${D}_{\inf}(h_r)={D}_{\inf}(h)$. In Appendix~\ref{app0dinfh}, we prove that finding $\mathcal{R}(D)$ is feasible if ${D}\in({D_{\inf}}(h),\infty)$. We show that in order to obtain a lower bound on ${\mathcal{R}}(D)$, one can minimize the directed information rate over an auxiliary coding scheme formed of LTI filters and an AWGN channel with feedback and delay. Inequalities and identities related to the delay-free version of this optimization derived in \cite{silva2016characterization} will be extended to the case with a constant channel delay. We start by deriving a lower bound on the average data rate $\mathcal{R}$ in the following theorem.
 \begin{theorem}\label{th2}
 	Suppose that the feedback system of Fig.~\ref{fig1} satisfies Assumptions~\ref{ass21} and \ref{ass23}. Then the average data rate $\mathcal{R}$ is lower bounded as follows:
 	\begin{equation}\label{eq182}
 	\mathcal{R}
 	\geq{{I_{\infty}^{(h)}}(y\to{u})}=\lim_{k \to \infty}
 	\tfrac{1}{k}
 	\Sigma_{i=0}^{k-1}I(u(i);{y^{i-h}}\mid{u^{i-1}}),
 	\end{equation}
 	 where ${{I_{\infty}^{(h)}}(y\to{u})}$ is the directed information rate across the forward channel from $y$ to $u$ with constant delay $h$ (see \cite[Definition~1]{derpich2013fundamental} for the formal definition).   	
 \end{theorem}
\begin{proof}
	Considering that ${l_k}=h$ holds at any $k\in{{\mathbb{N}}_0}$ for the NCS of Fig.~\ref{fig1}, we can conclude the claim immediately from Theorem~\ref{th5}.  
\end{proof}
The directed information rate in (\ref{eq182}) will be reduced if the involved signals are jointly Gaussian. This result is formalized by the following lemma.
\begin{lemma}\label{lemma1}
	Suppose that Assumptions~\ref{ass21} and \ref{ass23} hold for the NCS of Fig.~\ref{fig1}. Furthermore, consider $(x(0),w,u,y)$ as a jointly second-order set of random processes. Denote the Gaussian counterparts of $y$ and $u$ by $y_G$ and $u_G$, respectively, where $(x(0),w,u_G,y_G)$ are jointly Gaussian with the same first-and second-order (cross-) moments as $(x(0),w,u,y)$. Then ${{I_{\infty}^{(h)}}(y\to{u})\geq{I_{\infty}^{(h)}}({y_G}\to{u_G})}$.
\end{lemma}
\begin{proof}
Recall that ${l_k}=h$, $\forall{k}\in{{\mathbb{N}}_0}$, for the considered case with constant channel delay. The claim follows immediately from Lemma~\ref{lemma6}. 
\end{proof}
It can be implied from Lemma~\ref{lemma1} that by minimizing directed information rate over a scheme that renders $y$ and $u$ jointly Gaussian, one can obtain a lower bound on $\mathcal{R}(D)$. Now, we will show that the directed information rate can be stated in terms of power spectral densities of the involved processes if such signals meet certain stationarity conditions.
\begin{lemma}\label{lemma2}
	Suppose that $u$ is an SAWSS process with $|{{{\lambda}_{\min}}(C_{{u}_{1}^{n}})}|\geq{\mu}$, $\forall{n}\in{\mathbb{N}}$ where $\mu>0$. Moreover, assume that $u$ is jointly Gaussian and AWSS with the sensor output $y$. Then the directed information rate between $u$ and $y$ is expressed as 
	\begin{equation}\label{eq186}
	{{I_{\infty}^{(h)}}(y\to{u})}=\frac{1}{4\pi}\int_{-\pi}^{\pi} \log\big(\frac{{S_{\check{u}}}(e^{j\omega})}{{\sigma}_{\psi}^{2}}\big)d\omega,
	\end{equation}
	where $\psi$ represents a Gaussian AWSS process with independent samples defined by
	\begin{equation}\label{eq187}
	\psi(k)\triangleq{u(k)-\tilde{u}(k)},
	\tilde{u}(k)\triangleq{{\mathbf{E}}[u(k)\mid{{y^{k-h}},{u^{k-1}}}]}.
	\end{equation}
	Furthermore, ${S_{\check{u}}}$ denotes the steady-state power spectral density of $u$.
\end{lemma}  
\begin{proof}
Immediate from Lemma~\ref{lemma8} by noting that ${l_k}=h$ holds for the NCS of Fig.~\ref{fig1} at every time instant ${k}\in{{\mathbb{N}}_0}$.  
\end{proof}
It can be implied from Theorem~\ref{th2} and Lemma~\ref{lemma2} that the rate-performance pair yielded by any coding and control 
scheme satisfying Assumption~\ref{ass23} which renders the NCS of Fig.~\ref{fig1} SAWSS is attainable with a lower or equal rate if there exists a scheme that generates $(y,u)$ jointly Gaussian with
\begin{figure}[thpb]
	\centering
	\includegraphics[width=8cm]{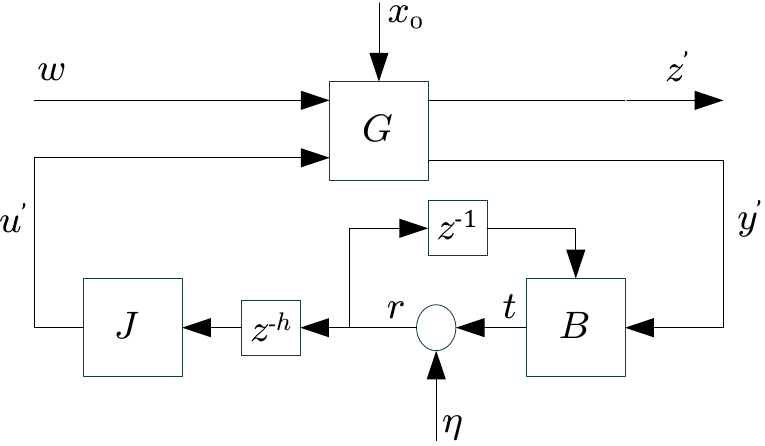}
	\caption{The LTI structure giving the lower bound in the constant delay case}
	\label{fig6}
\end{figure}
 $({x_0},w)$ while rendering the system SAWSS. Due to the Gaussianity of $(x_0,w)$ and the fact that the plant is LTI, a jointly Gaussian pair $(y,u)$ can be produced by a coding-control scheme comprised of LTI filters and an AWGN noise source. Such a scheme is depicted in Fig.~\ref{fig6}. The NCS of Fig.~\ref{fig6} satisfies all of the assumptions and conditions that hold for the system of Fig.~\ref{fig1}. However, the arbitrary mappings are replaced by proper LTI
filters $B$ and $J$ in the auxiliary feedback loop of Fig.~\ref{fig6}.

 In addition, for such an NCS, a delayed AWGN channel with noiseless one-sample-delayed feedback serves as communication channel. The coding-control scheme in the NCS of Fig.~\ref{fig6} is described via the following dynamics:
\begin{equation}\label{eq190}
u'=J{z^{-h}}r,\quad r=t+\eta,\quad t=B\mathrm{diag}\{{z^{-1}},1\} \begin{bmatrix}
r\\
y'
\end{bmatrix},
\end{equation}
where $\eta$ is a zero-mean white Gaussian noise with variance ${{\sigma}_{\eta}^{2}}$ and independent of $(x_0,w)$, and $B=[B_{r}\quad B_{y}]$. It should be emphasized that Assumption~\ref{ass21} holds for the initial states $x_0$, the plant $G$ and the disturbance input $w$ in the NCS of Fig.~\ref{fig6}. Furthermore, the initial states of the filters $B$ and $J$, and the delay blocks are deterministic. As the system depicted in Fig.~\ref{fig6} is a special case of the structure of Fig.~\ref{fig1}, we use apostrophes for presenting signals in Fig.~\ref{fig6} that have counterparts in the NCS of Fig.~\ref{fig1}.
\begin{theorem}\label{th3}
	If the NCS of Fig.~\ref{fig1} satisfies Assumption~\ref{ass21} and Assumption~\ref{ass23} and ${D}\in({D_{\inf}}(h),\infty)$ holds, then 
	\begin{equation}\label{eq191}
	{{\mathcal{R}}(D)}\geq{{{\vartheta}'}_u}(D)\triangleq\inf_{{{\sigma}_{z'}^{2}}\leq{D}}{\frac{1}{4\pi}\int_{-\pi}^{\pi} \log\big(\frac{{S_{u'}}(e^{j\omega})}{{\sigma}_{\eta}^{2}}\big)},
	\end{equation}
	where ${{\sigma}_{z'}^{2}}$ and ${S_{u'}}$ represent the steady-state variance of $z'$ and the steady-state power spectral density of $u'$ in Fig.~\ref{fig6}, respectively. Moreover, the feasible set for the optimization in (\ref{eq191}) is the set comprised of all LTI filters $B$ and the noise $\eta$ with ${{\sigma}_{{\eta}}^{2}}\in{{\mathbb{R}}_{+}}$ that render the system of Fig.~\ref{fig6} internally stable and well-posed with $J=1$.
\end{theorem}
\begin{proof}
See Appendix~\ref{app0th3}.
\end{proof}
Theorem~\ref{th3} implies that doing the optimization in (\ref{eq191}) over the auxiliary LTI system of Fig.~\ref{fig6}, with the AWGN channel and delay, will give a lower bound on the minimal data rate required to achieve a certain performance level in the arbitrary (possibly nonlinear and time-varying) structure of Fig.~\ref{fig1}. The following results show how the lower bound derived in (\ref{eq191}) can be simplified to a bound which is easier to compute. 
\begin{lemma}\label{lemma3}
	For the NCS of Fig.~\ref{fig6}, let describe ${\vartheta}_{r}^{'}$ by
	\begin{equation}\label{eq197}
	{\vartheta}_{r}^{'}(B,J,{{\sigma}_{{\eta}}^{2}})\triangleq\frac{1}{4\pi}\int_{-\pi}^{\pi} \log\big(\frac{{S_{r}}(e^{j\omega})}{{\sigma}_{\eta}^{2}}\big),
	\end{equation}
	where ${{\sigma}_{{\eta}}^{2}}\in{{\mathbb{R}}^{+}}$ is fixed and ${S_{r}}$ denotes the steady-state power spectral density of $r$. Moreover, suppose that the pair $(B,J)=({B_{1}},{J_{1}})$ renders the feedback loop of Fig.~\ref{fig6} internally stable and well-posed. Then for any $\rho>0$, there exists another pair with a biproper filter, say ${J_{2}}$, and a proper one, say ${B_{2}}$, that renders the system of Fig.~\ref{fig6} internally stable and well-posed, and preserves the steady-state power spectral density of $z'$ in a way that the following holds:
	\begin{IEEEeqnarray}{rl}\label{eq198}
		\begin{split}
			{\vartheta}_{r}^{'}({B_{1}},{J_{1}},{{\sigma}_{{\eta}}^{2}})={\vartheta}_{r}^{'}({B_{2}},{J_{2}},{{\sigma}_{{\eta}}^{2}})={\frac{1}{2}}\log(1+{\frac{{\sigma}_{t}^{2}}{{\sigma}_{\eta}^{2}}}){{\mid}_{(B,J)=({B_2},{J_2})}´}-\rho.
		\end{split}	
	\end{IEEEeqnarray}   
\end{lemma}
\begin{proof}
See Appendix~\ref{app0lemma3}.
\end{proof}
Intuitively speaking, the results of Theorem~\ref{th3} and Lemma~\ref{lemma3} imply that $\mathcal{R}(D)$ can be bounded from below by a logarithmic term as in (\ref{eq198}) which is a function of channel SNR in the NCS of Fig.~\ref{fig6}. Such an intuition will assist us with deriving a lower bound which is computationally appealing in the following corollary.
\begin{corollary}\label{cor1}
	Take the feedback loop of Fig.~\ref{fig1} into account as an NCS that satisfies Assumptions~\ref{ass21} and~\ref{ass23}. Then for every ${D}\in({D_{\inf}}(h),\infty)$, the following holds:
	\begin{equation}\label{eq23}
	{{\mathcal{R}}(D)}\geq{\frac{1}{2}}\log(1+{{\varphi'}(D)}), \quad{{\varphi'}(D)}\triangleq{{\inf_{{{\sigma}_{z'}^{2}}\leq{D}}}{\frac{{\sigma}_{t}^{2}}{{\sigma}_{\eta}^{2}}}}, 
	\end{equation}
	in which ${{\sigma}_{t}^{2}}$ and ${{\sigma}_{z'}^{2}}$ symbolize the steady-state variances of $t$ and $z'$ in the auxiliary system of $Fig.~\ref{fig6}$, respectively. For the optimization problem in (\ref{eq23}), a candidate solution is an LTI filter pair $(B,J)$ together with noise variance ${{\sigma}_{{\eta}}^{2}}\in{{\mathbb{R}}^{+}}$ that cause the system in Fig.~\ref{fig6} to become internally stable and well-posed.  
\end{corollary}
\begin{proof}
See Appendix~\ref{app0cor1}.
\end{proof}
\section{Upper bound problem in the presence of constant delay}
In this section, we show that for any ${D}\in({D_{\inf}}(h),\infty)$, one can always find a scheme that guarantees attaining ${{{\sigma}_{z}^{2}}\leq{D}}$ with an average data rate which has a distance of about $1.254$ bits per sample from the theoretical lower bound. For such a scheme, we propose a design approach which utilizes the filters that together with an AWGN with feedback and delay, render the directed information rate over the channel equal to the lower bound on ${{\mathcal{R}}(D)}$.
\begin{definition}
	We call a coding scheme with input-output relationship as in (\ref{eq24})-(\ref{eq271}) in the constant channel delay case linear if and only if its dynamics can be restated as follows:
	\begin{equation}\label{eq38}
	u=J{z^{-h}}r,\quad r=t+\eta,\quad t=B\mathrm{diag}\{{z^{-1}},1\} \begin{bmatrix}
	r\\
	y
	\end{bmatrix},
	\end{equation}
where $B=[B_{r}\quad B_{y}]$ and $J$ are proper LTI filters with deterministic initial condition. Moreover, $\eta$ represents a zero-mean i.i.d random sequence independent of $(x_0,w)$. The initial state of the one-step-delay feedback channel  is assumed to be deterministic. 
\end{definition}
The realization of linear source coding schemes can be carried out by using entropy-coded dithered quantizers (ECDQs) together with LTI filters. First, implementing an ECDQ causes the following relationship between $({u_q},{y_q})$ in (\ref{eq7}) and (\ref{eq25}), and $(r,t)$ in (\ref{eq38}):
\begin{IEEEeqnarray}{rcl}\label{eq39}
	\begin{split}
		{{y_{\mathcal{E}}}(k)}&={{\mathcal{F}}_q}(t(k)+d(k))\\
		{y_q}(k)&={\mathcal{O}_k}({{y_{\mathcal{E}}}(k)},d(k))\\
		{{u_{\mathcal{D}}}(k)}&={\mathcal{O}_{k-h}^{-1}}({u_q}(k),d(k-h))\\
		{r_h}(k)&={{u_{\mathcal{D}}}(k)}-d(k-h),
	\end{split}	
\end{IEEEeqnarray}
\begin{figure}[thpb]
	\centering
	\includegraphics[width=8cm]{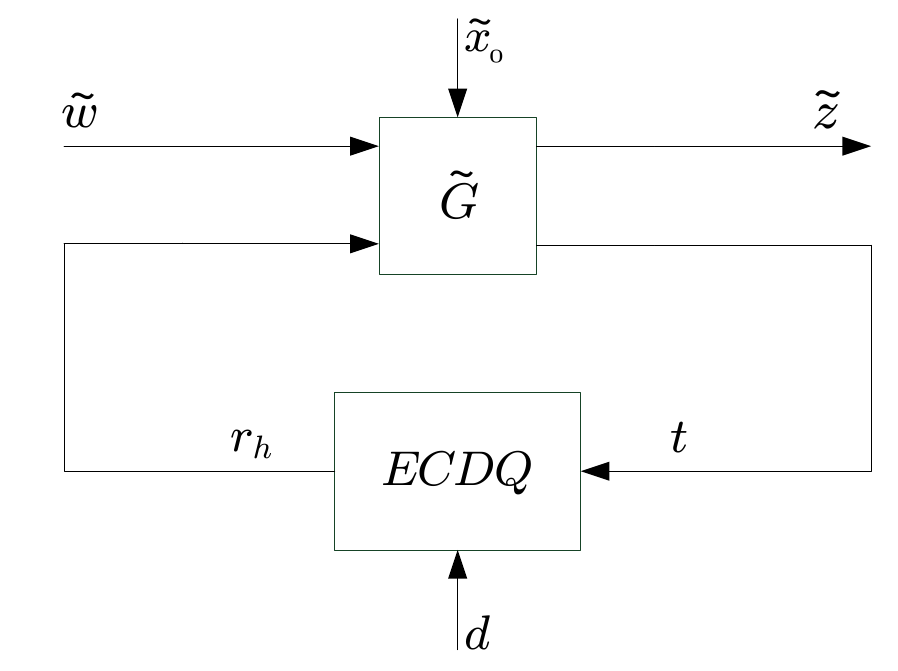}
	\caption{ECDQ setup in the feedback path}
	\label{fig8}
\end{figure}
in which by ${{\mathcal{F}}_q}$, we denote a uniform quantizer with resolution $\Delta\in\mathbb{R}^{+}$, ${{\mathcal{F}}_q}:\mathbb{R}\to\{i\Delta;i\in\mathbb{Z}\}$. Additionally, $d(k)$ represents a dither signal whose access are provided to both encoder and decoder. The mapping ${\mathcal{O}_k}$ and its complementary ${\mathcal{O}_{k}^{-1}}$ formalize entropy coding for the lossless parts at the encoder and decoder, respectively. The following lemma presents an interesting property of ECDQs when being set up in an LTI feedback loop. 
 \begin{lemma}\label{lemma4}
	Consider the feedback loop depicted in Fig.~\ref{fig8} and suppose that the plant $\tilde{G}$ is described by a proper real rational transfer-function matrix in which the transfer function from $r_h$ to $t$ is scalar and strictly proper. For such a system, assume that the input-output relationship of the ECDQ in the feedback path is given by (\ref{eq39}) with finite and positive quantization step size $\Delta$. Moreover, take the disturbance $\tilde{w}$ into account as a white noise process jointly second-order with $\tilde{x}_0$, the initial state of $\tilde{G}$. Then if the dither $d$ is an i.i.d process with a uniform distribution over $(-\Delta/2,\Delta/2)$ and independent of ($\tilde{w},\tilde{x}_0$), the error $r-t$ is i.i.d, uniformly distributed over $(-\Delta/2,\Delta/2)$ and independent of ($\tilde{w},\tilde{x}_0$). 
\end{lemma} 
\begin{proof}
	See Appendix~\ref{app0lemma4}.
\end{proof}
It can be implied from above that combining the LTI filters in (\ref{eq38}) with the ECDQ in (\ref{eq39}) in a setting as depicted in Fig.~\ref{fig9} will lead to a linear coding scheme for the NCS of Fig.~\ref{fig1} as long as $d(k)$ 
\begin{figure*}[thpb]
	\centering
	\includegraphics[width=15cm]{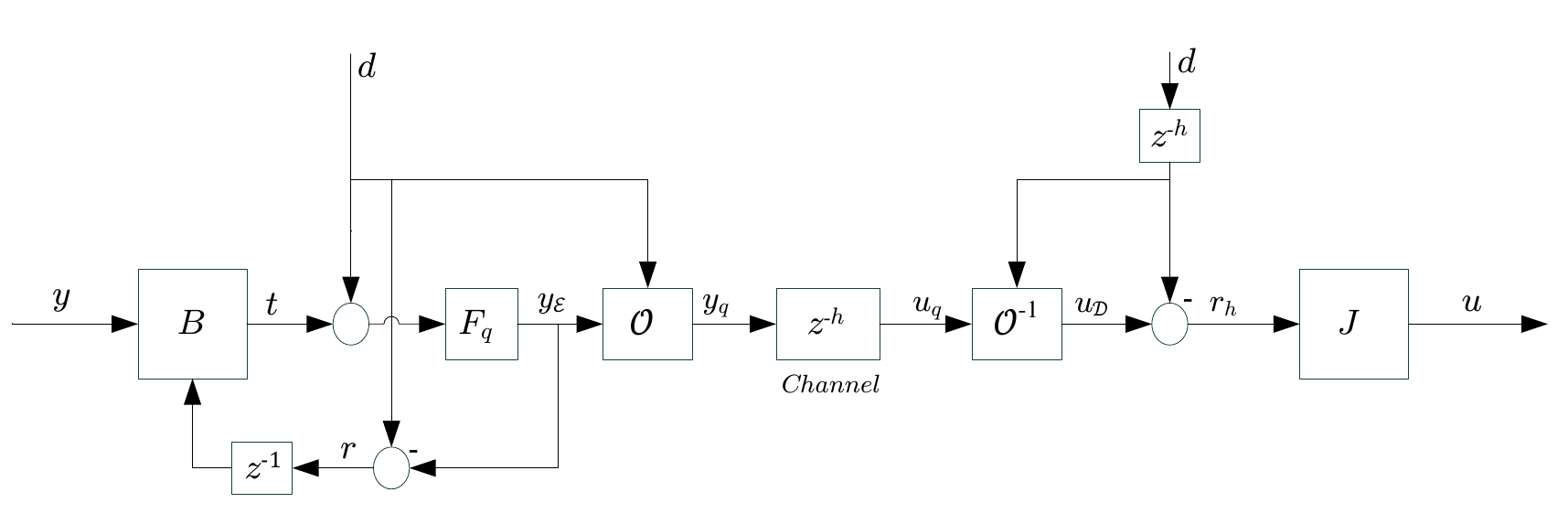}
	\caption{The proposed ECDQ-based linear coding scheme}
	\label{fig9}
\end{figure*}
meets the same criteria as for the dither in Lemma~\ref{lemma4}. If so, the obtained coding scheme is called a linear ECDQ-based coding scheme. If such a scheme is implemented on the feedback path of the main NCS of Fig.~\ref{fig1}, the average data rate is bounded from above by a certain value which is shown in the following lemma.
\begin{lemma}\label{lemma5}
	Suppose that Assumption~\ref{ass21} holds for the NCS of Fig.~\ref{fig1}. Then the existence of an ECDQ-based linear source-coding scheme rendering the NCS of Fig.~\ref{fig1} SAWSS is certified in such a way that the average data rate satisfies
	\begin{equation}\label{eq41}
	\mathcal{R}<{\frac{1}{2}}\log\left(1+{\frac{{\sigma}_{t}^{2}}{{\sigma}_{\eta}^{2}}}\right)+{\frac{1}{2}}\log\left(\frac{2{\pi}e}{12}\right)+\log{2}.
	\end{equation}
	In (\ref{eq41}), the variance of the quantization error (noise) of the ECDQ-based linear source-coding scheme is set as ${{\sigma}_{\eta}^{2}}={{\Delta}^{2}}/12$. Moreover, ${{\sigma}_{t}^{2}}$ represents the steady-state variance of the signal $t$ in (\ref{eq38}). 
\end{lemma}
\begin{proof}
	See Appendix~\ref{app0lemma5}.
\end{proof}
Now, through the following theorem, we use the result of Lemma~\ref{lemma5} to show that utilizing ECDQ-based linear coding schemes can lead to an upper bound on the desired minimal average data rate $\mathcal{R}(D)$. 
\begin{theorem}\label{th4}
	Let Assumption~\ref{ass21} hold for the closed-loop system of Fig.~\ref{fig1}. Then for each $D\in{(D_{\inf}(h),\infty)}$, one can always find an ECDQ-based linear source-coding scheme satisfying Assumption~\ref{ass23} and rendering the feedback loop of Fig.~\ref{fig1} SAWSS in such a way that ${{{\sigma}_{z}^{2}}\leq{D}}$ is resulted and the average data rate is bounded as  
	\begin{equation}\label{eq42}
	\mathcal{R}<{\frac{1}{2}}\log\left(1+{{\varphi'}(D)}\right)+{\frac{1}{2}}\log\left(\frac{2{\pi}e}{12}\right)+\log{2},
	\end{equation}
where the definition of ${{\varphi'}(D)}$ is given in (\ref{eq23}).	 
\end{theorem}
\begin{proof}
See Appendix~\ref{app0th4}.
\end{proof}
In the following remark, we state how the upper bound derived in Theorem~\ref{th4} can be considered as an upper bound on ${\mathcal{R}_a}(D)$ in the case of random channel delay.
\begin{remark}
The upper bound in (\ref{eq42}) will be an upper bound on ${\mathcal{R}_a}(D)$ in the random channel delay case if coding and control schemes are linear ECDQ-based schemes designed as in the proof of Theorem~\ref{th4} for the delay ${h_{\max}}$ where the decoder-controllers have buffers installed at their inputs sending only ${y_q}(k-h_{\max})$ for prcessing at each time instant $k\in{\mathbb{N}_0}$. 
	\begin{figure}[h]
	\centering
	\includegraphics[width=8cm]{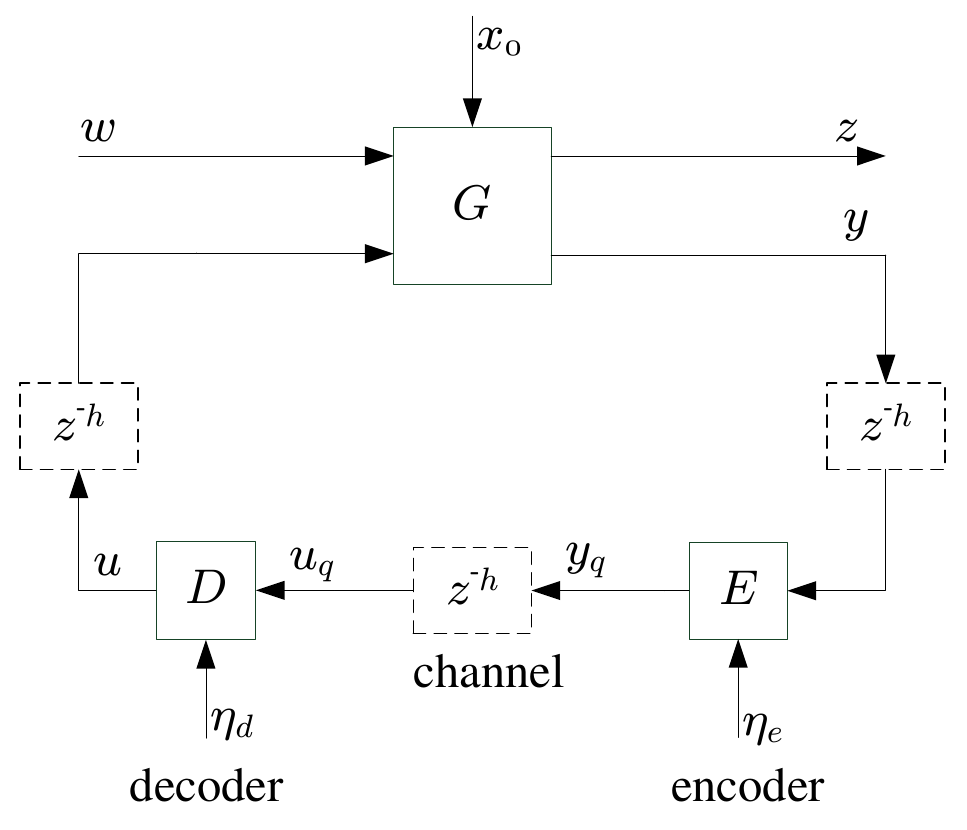}
	\caption{Three possible locations for the delay component in the case with constant channel delay}
	\label{fig14}
\end{figure}
Clearly, this is due to the fact that at every time step $k\in{\mathbb{N}_0}$, ${y_q}(k-h_{\max})$ is available at the decoder. Such an upper bound does not seem to be tight since imposing a delay of $h_{\max}$ steps on transmitted data is actually a worst-case scenario.     
\end{remark}
The bounds derived in this section and the previous section limit the desired average data rate $\mathcal{R}(D)$ in the NCS of Fig.~\ref{fig1}. In this system, the constant delay is induced by the digital communication channel between the encoder-controller and the decoder-controller. One concern is the effect of delay location on the derived bounds. The following lemma takes a step in addressing this issue by showing how the system signals change when the time delay block is moved to a different location in the feedback loop of Fig.~\ref{fig1}.  
\begin{lemma}\label{lemma01}
	Consider the NCS of Fig.~\ref{fig1} and two other systems each of which yielded by moving the delay component in the NCS of Fig.~\ref{fig1} to either the measurement path (between the sensor and the encoder-controller) or the actuation path (between the decoder-controller and the plant). Fig.~\ref{fig14} depicts the locations where the time delay occurs in these cases. 
	Then systems are not necessarily equivalent across the cases if the only difference between them is the delay location. 
	However, the equivalence can be assured by allowing the side information to change across the cases.   
\end{lemma}
\begin{proof}
	See Appendix~\ref{app0lemma01}.
\end{proof}
\section{numerical Simulation}
 Take the following transfer function into account as the model describing the generalized plant $G$ in the NCS of Fig.~\ref{fig1}:
\begin {equation}\label{eq30}
 z=	\frac{0.165}{(z-2)(z-0.5789)}(w+u),\quad y=z,
\end{equation}
Let us set the disturbance signal $w$ and initial states $x_0$ in such a way that Assumption~\ref{ass21} is satisfied. We calculated lower and upper bounds on $\mathcal{R}(D)$ as derived in (\ref{eq23}) and (\ref{eq42}). For computing these bounds, we made use of the equivalence between the NCSs of Fig.~\ref{fig6} and Fig.~\ref{fig10}, shown in the proof of Lemma~\ref{lemma3}, in that we adopted the method in \cite{silva2016characterization} which solves  SNR-performance optimization problems similar to the one defining ${{\varphi'}(D)}$ for such systems as the NCS of Fig.~\ref{fig10}. The bounds are computed for three different values of channel delay, $h=\{0,1,2\}$, with respect to $D$ varying over a range from ${D_{\inf}}(h)$ to $50$ for each $h$. Moreover, we designed actual linear ECDQ-based coding schemes, and for each selected $D$ in the latter interval, we simulated the NCS of Fig.~\ref{fig1}. To do so, we utilized the filters giving the lower bound on $\mathcal{R}(D)$ according to the procedure suggested in \cite[Theorem~5.1]{silva2016characterization}. The results are demonstrated in Fig.~\ref{fig13}. In this figure,
		\begin{figure}[thpb]
	\centering
	\includegraphics[width=11cm]{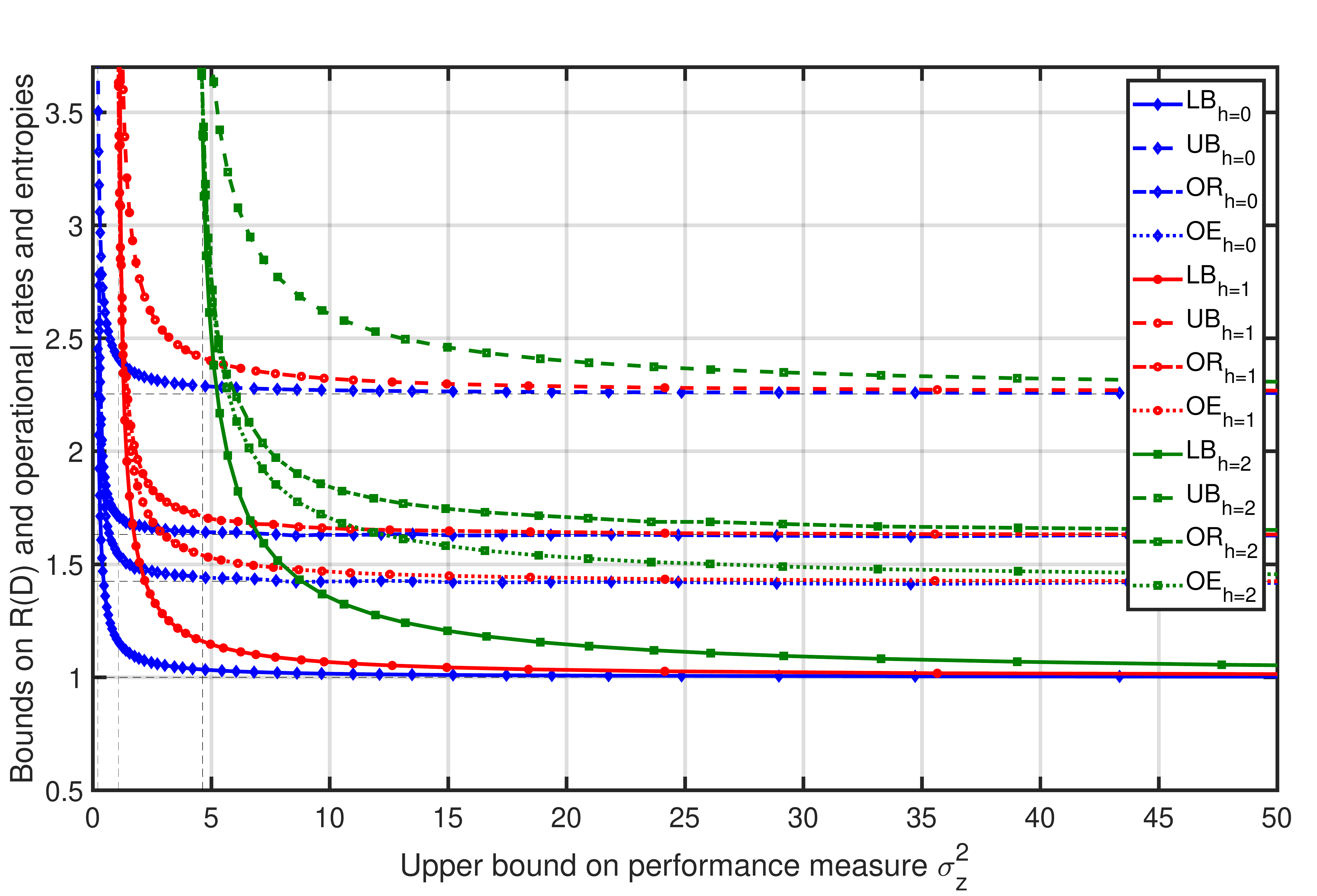}
	\caption{Bounds on $\mathcal{R}(D)$ in (\ref{eq23}) and actual data rates and entropies for different values of time delay $h$}
	\label{fig13}
\end{figure}
 the curves referred to as LB and UB present the lower and upper bounds on $\mathcal{R}(D)$, respectively. We can compare ${D_{\inf}}(h)$ among cases with different values of channel delay as well. As shown, greater ${D_{\inf}}(h)$ is associated with larger channel delay, as expected according to \cite{hashikura2014h}. Evaluating how the bounds change in response to changes in the delay is one of the main purposes of this simulation study. We can observe from the bounds plotted in Fig.~\ref{fig13} that when $D$ is fixed, increasing the delay will enlarge the bounds on $\mathcal{R}(D)$. In other words, the greater the delay is, the higher average data rate is to be used in order to achieve a fixed quadratic performance level. Moreover, Fig.~\ref{fig13} shows that the lower (upper) bound curves converge to the minimum data rate required for mean square stability as $D$ grows larger. From\cite{johannesson2011control}, we know that the minimal data rate guaranteeing stabilizability of the NCS of Fig.~\ref{fig6} is only a function of unstable poles of the plant $G$. On the other hand, we use the equivalent system of Fig.~\ref{fig10} for the purpose of calculating bounds. So the observation with convergence of bounds to the minimal data rate needed for stability comes from the fact that incorporating time delay into the model of the plant $G_a$ will not affect its unstable poles.

Simulation results are illustrated in Fig.~\ref{fig13} as well. The curves referred to as OR and OE present the average data rates and entropies achieved by using actual linear coding schemes. Furthermore, $10^6$-sample-long realizations have been considered for the dither. The coding task in all utilized schemes is done by memory-less Huffman coders which do not take the past information of the dither into account as prior knowledge for coding. In addition to the average data rate, the entropy of the output of the quantizer has been estimated for the aforementioned setup. The gap of around $0.4$ bits per sample between the measured entropy and the lower bound indicates that for each $h\in\{0,1,2\}$, $0.4$ bits per sample of the gap between the actual rates and lower bound is caused by replacing the AWGN with uniform dither and the remainder $0.25$ bits per sample corresponds to sample-by-sample coding. It can be observed that the actual rates and entropies have the same properties as the properties of bounds mentioned in the previous paragraph. The most prominent property is related to the behaviour of the achieved rates and entropies as a function of channel delay, i.e., for a system with greater time delay in the channel, higher rates are required to guarantee quadratic performance requirements.
\section{conclusions}
In this paper, the trade-off between average data rate and performance in networked control systems has been studied. Two setups have been investigated, each of which incorporates an LTI plant with Gaussian disturbance and initial states, and scalar control input and sensor output. Moreover, both of them have causal, but otherwise arbitrary, mappings on their feedback paths which are responsible for coding and control. The only difference between the two considered systems is the model of the channel that carries out data transmission between the encoder-controller and the decoder-controller. In one case, the digital communication channel is SIMO and information to be exchanged are exposed to random delay. In the other system, the channel is error-free as well but it is SISO and imposes constant delay on transmitted data. For the case with random channel delay, we considered notions for rate and performance which show the average behaviour of the system over all realizations of the delay. We have shown that for such a setup, data rate is lower bounded by average directed information rate from the sensor output to control input, and if $y$ and $u$ are jointly Gaussian, the average directed information rate would be lowest. Moreover, we have shown that when $y$ and $u$ satisfy certain stationarity assumptions, the average directed information rate between them is a function of the average power spectral densities of these signals over all realizations of the channel delay. We have shown that infimum value of this function over all arbitrary coders and controllers that cause system signals have those Gaussianity and staionarity properties lower bounds the infimum average data rate required to attain a prescribed quadratic performance.

For the constant delay case, which is a special case of the system with random channel delay, we approximated (by deriving bounds) the minimal average data rate that certifies attaining a certain performance level. Employing the fundamental information inequalities and identities derived for the random delay case, we showed that this desired minimal average data rate is bounded from below when coder-controllers and the channel behave as a concatenation of proper LTI filters and an AWGN channel with feedback and delay. Then we showed that by approximating such schemes with simply implementable linear ECDQ-based coding schemes, one can achieve any (legitimate) performance level by actual rates which are at most $1.254$ bits per sample higher than the lower bound. The results illustrated through the simulation show that bounds and empirical rates are increasing functions of channel delay for a fixed performance level. It means larger delay in the channel necessitates higher minimal average data rate that is needed for achieving a certain level of quadratic performance.

Future research will concern with finding closed-form solution for the lower and upper bound problems in the case of random channel delay, finding analytic expression for the desired infimum data rate, deriving lower and upper bounds with shorter gap between them, plants with model uncertainties and vector quantization.  
           
\appendices

\section{Invertibility of the decoder}\label{appb}
\begin{lemma}
	Consider a coding scheme described through (\ref{eq24})-(\ref{eq271}) that has a non-invertible decoder, and let $\check{R}(k)$ be defined as $\check{R}(k)\triangleq{H({{y_{\mathcal{E}}}(k)}\mid{{y_{{\mathcal{E}}f}(k)},{\eta_{o}^{k}}})}$. For such scheme, assume that $u(k)={u_0}(k)$ and ${\check{R}}_{f}(k)={\check{R}}_{f0}(k)$, $\forall{k}\in{\mathbb{N}_0}$, where ${{\check{R}}_{f}}(k)\triangleq{{[\check{R}}(i))]_{i\in{S(k)}}^T}$. Then there exists another coding scheme constructing the control input $u(k)={u_0}(k)$ with an invertible decoder in such a way that ${\check{R}}_{f}(k)\leq{\check{R}}_{f0}(k)$, $\forall{k}\in{\mathbb{N}_0}$. 
\end{lemma}
\begin{proof}
	Suppose that mappings in (\ref{eq26})-(\ref{eq271}) represent a non-invertible decoder at time $k$ in a way that upon knowledge of ${\eta}_{d}^{i}$ and ${S^i}$, perfect reconstruction of $u_q^i$ from $u^i$ has been possible for all $i\leq{k-1}$. Then there exist ${u_{\mathcal{D}1}},{u_{\mathcal{D}2}}\in{\mathcal{A}_{s}}$ such that ${u(k)}={\mathcal{D}_k}({u_{\mathcal{D}1}},{u_{\mathcal{D}}^{k-1}},{\eta_{d}^{k}})={\mathcal{D}_k}({u_{\mathcal{D}2}},{u_{\mathcal{D}}^{k-1}},{\eta_{d}^{k}})$. Let ${S_{1}}$ and ${S_{2}}$ be associated with ${u_{\mathcal{D}1}}$ and ${u_{\mathcal{D}2}}$ respectively. Two possible cases can occur. In the first case, ${S_{1}}$  and ${S_{2}}$ are unequal, i.e., ${S_{1}}\neq{S_{2}}$. Since $S$ is known at the decoder at each time step, this case does not contradict the invertibility. That is due to the fact that the knowledge of ${S}$ would determine whether $u(k)$ is caused by ${u_{\mathcal{D}1}}$ or ${u_{\mathcal{D}2}}$. However, the situation is not the same in the case where ${S_{1}}={S_{2}}$. Since both ${u_{\mathcal{D}1}}$ and ${u_{\mathcal{D}2}}$ are vector-valued variables, ${u_{\mathcal{D}1}}\neq{u_{\mathcal{D}2}}$ means that at least one entry of ${u_{\mathcal{D}1}}$ is not equal to the entry with the same dimension in ${u_{\mathcal{D}2}}$. The corresponding elements of ${u_{\mathcal{D}1}}$ and ${u_{\mathcal{D}2}}$ that are not equal to each other are denoted by pairs $({u_{\mathcal{D}1j}},{u_{\mathcal{D}2j}})$, $1\leq{j}\leq{m}$, $1\leq{m}\leq{{h_{\max}}+1}$. Since ${S_{1}}={S_{2}}$, for each $j$, both ${u_{\mathcal{D}1j}}$ and ${u_{\mathcal{D}2j}}$ have been exposed to the same delay, say $h_j$, $0\leq{h_j}\leq{h_{\max}}$. So if we denote the output of the lossy encoder that corresponds to ${u_{\mathcal{D}nj}}$ by ${y_{\mathcal{E}nj}}$, then ${u_{\mathcal{D}nj}}(k)={y_{\mathcal{E}nj}}(k-{h_j})$. It should be noted that $n$ is a positive integer which is at most equal to the size of the set ${\mathcal{A}_{s}}$. Let ${p_{nj}}$ represent the conditional probability of having ${y_{\mathcal{E}nj}}$ at the encoder given $({y_{{\mathcal{E}}f}}({k-{h_j}}),{\eta_{o}^{k-{h_j}}})$ at time $k-{h_j}$. The encoder-decoder set $(\bar{\mathcal{E}},\bar{\mathcal{D}})$ can be defined with exactly the same properties as $({\mathcal{E}},{\mathcal{D}})$ but different from it in the sense that $\bar{\mathcal{E}}$ outputs only ${y_{\mathcal{E}1j}}$ at time $k-{h_j}$ with probability ${p_{1j}}+{p_{2j}}$. This means having only ${u_{\mathcal{D}1j}}$ at time $k$ as decoder input instead of receiving either ${u_{\mathcal{D}1j}}$ or ${u_{\mathcal{D}2j}}$. Let us define ${t_j}\triangleq{k-{h_j}};k\geq{h_j}$. Then
	\begin{align}
	\begin{split}
	{\check{R}(t_j)\mid}_{({\mathcal{E}},{\mathcal{D}})}&={\check{R}_0}(t_j)\\
	&\myeqaa-\sum_{n\notin\{1,2\}}{{p_{nj}}\ln{p_{nj}}-{p_{1j}}\ln{p_{1j}}-{p_{2j}}\ln{p_{2j}}}\\
	&\myeqab-\sum_{n\notin\{1,2\}}{{p_{nj}}\ln{p_{nj}}-({p_{1j}}+{p_{2j}})\ln({p_{1j}}+{p_{2j}})}\\
	&\myeqac{\check{R}(t_j)\mid}_{({\bar{\mathcal{E}}},{\bar{\mathcal{D}}})}
	\end{split} 
	\end{align}
	in which $(aa)$ results from the definition of entropy and $\check{R}(k)$, $(ab)$ can be concluded based on the fact that the function $-\ln({p_{nj}})$ is monotonically decreasing, and  $(ac)$ follows from the definition of $\check{R}(k)$ for the scheme $(\bar{\mathcal{E}},\bar{\mathcal{D}})$. So ${\check{R}(t_j)\mid}_{({\bar{\mathcal{E}}},{\bar{\mathcal{D}}})}\leq{\check{R}(t_j)\mid}_{({\mathcal{E}},{\mathcal{D}})}$ for ${t_j}\geq{0}$, and consequently ${\check{R}}_{f}(k)\leq{\check{R}}_{f0}(k)$.

	The above procedure can be iterated for evey pair with the same characteristics as $({u_{\mathcal{D}1}} ,{u_{\mathcal{D}2}})$ to make sure that there are no two inputs of the reproduction decoder  mapped into one identical $u(k)$ at time instant $k$. Such iteration will then yield an invertible decoder. In other words, when the pair $(\bar{\mathcal{E}},\bar{\mathcal{D}})$ is used, knowing $({u^{i}},{\eta_{d}^{i}},{S^i})$ is equivalent to knowing $({u_{\mathcal{D}}^{i}},{\eta_{d}^{i}},{S^i})$ with $u(i)={u_0}(i)$ and ${\check{R}}_{f}(i)\leq{\check{R}}_{f0}(i)$, $\forall{i}\leq{k}$. Our main claim now follows  by repeating the above for every ${k}\geq{0}$.  
\end{proof}
\section{Proofs}\label{app0}
\subsection{Feasibility proof for ${D_{\inf}(h)}$, ${{\vartheta}_{u}^{'}(D)}$ and $\varphi'(D)$}\label{app0dinfh}
Suppose that in the standard architecture depicted in Fig.~\ref{figa1}, $G$, $x_0$ and $w$ satisfy Assumption~\ref{ass21} and $K$ follows $u(k)={{\mathcal{K}}_{k}}({y}^{k-h})$. Considering the Gaussianity of $x_0$ and $w$ and the fact that $G$ is LTI, we can imply from some results in \cite{aastrom2012introduction} that:
\begin{equation}\label{eqkfis}
{D}_{\inf}(h)=\inf_{K\in{\kappa}}{\sigma}_{z}^{2},
\end{equation}  
in which ${\sigma}_{z}^{2}$ denotes the steady-state variance of output $z$ and ${\kappa}$ is the set of all proper LTI filters which render the system of Fig.~\ref{figa1} internally stable and well-posed. The assumptions considered for $G$ guarantee that finding ${D}_{\inf}(h)$ is feasible.
\begin{figure}[thpb]
	\centering
	\includegraphics[width=8cm]{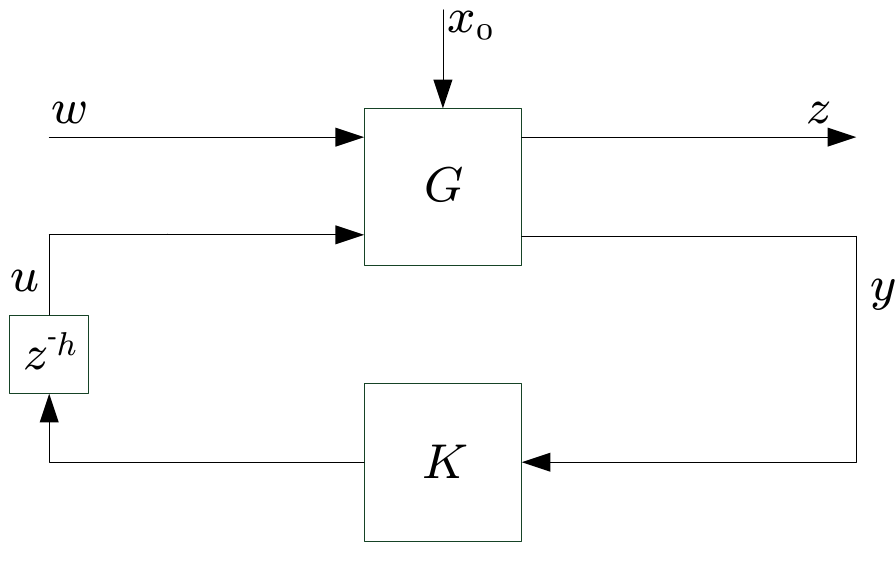}
	\caption{Standard feedback loop over which ${D_{\inf}}(h)$ is defined}
	\label{figa1}
\end{figure} 
Since ${D}_{\inf}(h)$ can be obtained, for every $\zeta\in(0,{D-{D_{\inf}(h)})}$, there exists ${K_1}\in{\kappa}$ which gives ${\sigma}_{z_{1}}^{2}\triangleq{{\sigma}_{z}^{2}}\mid_{K=K_1}\leq{{D_{\inf}(h)}+\zeta<D}$ for the system of Fig.~\ref{figa1}. Applying ${K_1}$ to this system results in a stable setting which is a special case of the NCS depicted in Fig.~\ref{fig6} with $J=1$ and $r=t={K_1}y'$ where the steady-state variance of $t$, ${{\sigma}_{t}^{2}}={{\sigma}_{t_{1}}^{2}}$, is finite. Therefore, since ${{K_1}\in{\kappa}}$, it can bring internal stability and well-posed-ness to the feedback loop of Fig.~\ref{fig6} in the presence of any additive noise $\eta$ with steady-sate variance ${{\sigma}_{{\eta}}^{2}}\in{{\mathbb{R}}^{+}}$. 
So ${\sigma}_{z'}^{2}={\sigma}_{z_{1}}^{2}+{{\chi}_{z}}{{\sigma}_{\eta}^{2}}$ and ${\sigma}_{t}^{2}={{\sigma}_{t_{1}}^{2}}+{{\chi}_{t}}{{\sigma}_{\eta}^{2}}$ can be concluded, when taking $\eta$ into account as an AWGN with finite variance ${{\sigma}_{{\eta}}^{2}}$ for the system of Fig.~\ref{fig6}. It should be noted that  ${{\chi}_{t}},{{\chi}_{z}}\geq{0}$ depend only on $K_1$. 
Now by choosing $\zeta=({D-{D_{\inf}}(h)})/3$ and the variance ${{\sigma}_{{\eta}}^{2}}=({D-{D_{\inf}}(h)})/(3{{\chi}_{z}})$ for the AWGN, there exists ${K_1}\in{\kappa}$ rendering the NCS of Fig.~\ref{fig6} internally stable and well-posed in a way that ${{\sigma}_{z'}^{2}{\mid}_{(B,J,{\sigma}_{\eta}^{2})=(K_1,1,{\sigma}_{\eta}^{2})}}\leq{{D_{\inf}(h)}+{\frac{2}{3}}(D-{D_{\inf}(h)})}<D$. Then the following can be obtained for the structure of Fig.~\ref{fig6}:
\begin{equation}\label{eq31}
{{{\frac{{\sigma}_{t}^{2}}{{\sigma}_{\eta}^{2}}}}{\mid}_{(B,J,{\sigma}_{\eta}^{2})=(K_1,1,{\sigma}_{\eta}^{2})}}={\frac{{3{{\sigma}_{t_1}^{2}}{{\chi}_{z}}}}{{D-{D_{\inf}}(h)}}}+{{\chi}_{t}}<\infty.
\end{equation}
So considering Jensen's inequality and concavity of logarithm, we can deduce that the problem of finding ${{\vartheta}_{u}^{'}(D)}$ in (\ref{eq191}) is feasible for every $D>{D_{\inf}}(h)$. The feasibility of the problem of finding $\varphi'(D)$ in (\ref{eq23}) is inferred immediately from (\ref{eq31}) for any $D>{D_{\inf}}(h)$.
\subsection{Proof of Theorem~\ref{th3}}\label{app0th3}
Due to the validity of $D>{{D}_{\inf}}(h)$, one can always find at least one coding-control pair, say $\hat{{E}}$ and $\hat{{D}}$, that while satisfying Assumption~\ref{ass23}, renders the NCS of Fig.~\ref{fig1} SAWSS in such a way that ${{{\sigma}_{\hat{z}}^{2}}\leq{D}}$ and
\begin{IEEEeqnarray}{rcl}\label{eq192}
	\begin{split}
		\mathcal{R}
		\geq{{I_{\infty}^{(h)}}(\hat{y}\to\hat{u})}
		\geq{I_{\infty}^{(h)}}({\hat{y}_G}\to\hat{u}_G)=\frac{1}{4\pi}\int_{-\pi}^{\pi} \log\big(\frac{{S_{\breve{u}}}(e^{j\omega})}{{\sigma}_{{\hat{\psi}}_G}^{2}}\big)d\omega.
	\end{split}	
\end{IEEEeqnarray} 
In (\ref{eq192}), processes $\hat{z}$, $\hat{y}$ and $\hat{u}$ are the counterparts of $z$, $y$ and $u$ in Fig.~\ref{fig1}, respectively. Moreover, the inequalities and identities in (\ref{eq192}) stem from Theorem~\ref{th2} if conditions in Lemma~\ref{lemma1} and Lemma~\ref{lemma2} are satisfied. Therefore, $({\hat{y}_G},{\hat{u}_G})$ are jointly Gaussian counterparts of $(\hat{y},\hat{u})$ as in Lemma~\ref{lemma1} and ${S_{\breve{u}}}$ represents the steady-state power spectral density of $\hat{u}_G$ as in Lemma~\ref{lemma2}. The pair $({\hat{y}_G},{\hat{u}_G})$ with conditions stated in Lemma~\ref{lemma1} can be generated by a scheme which certifies ${{\sigma}_{\hat{z}_G}^{2}}\in({D_{\inf}}(h),\infty)$ and is comprised of linear filters with a unit-gain noisy channel and delay $h$ as follows:
\begin{equation}\label{eq194}
{\hat{u}_{G}}(k)={L_k}({\hat{y}}_{G}^{k-h},{\hat{u}}_{G}^{k-1})+{\hat{\psi}_{G}}(k-h),\qquad k\in{{\mathbb{N}}_0},
\end{equation}    
in which ${\hat{\psi}_{G}}(k)$ denotes a Gaussian noise with zero mean and independent of $({\hat{y}}_{G}^{k},{\hat{u}}_{G}^{k-1})$. Since ${L_k}$ is a linear and causal mapping, we can redescribe ${\hat{u}_{G}^{k}}$ as
\begin{equation}\label{eq195}
{\hat{u}_{G}^{k}}={Q_k}{\hat{\psi}_{G}^{k-h}}+{P_k}{\hat{y}}_{G}^{k-h},\qquad k\in{{\mathbb{N}}_0}.
\end{equation} 
It follows from causality in (\ref{eq195}) that $\forall{k}\in{\mathbb{N}}$, $B_k$ and $G_k$ are lower triangular matrices with $B_{k-1}$ and $G_{k-1}$ on the top left corners. This together with the fact that $({\hat{y}_G},\hat{u}_G)$ are jointly SAWSS allow us to conclude that based on transitivity of asymptotic equivalence for products and sum of the matrices in \cite{gray2006toeplitz}, the sequences $\{Q_k\}$ and $\{P_k\}$ are asymptotically equivalent to sequences of lower triangular Toeplitz matrices. Furthermore, using $L_k$ as in (\ref{eq194}) will bring internal stability and well-posed-ness to the corresponding NCS. 
Now let us set $J=1$ and $B$ as a concatenation of linear filters with the same behaviour as steady-state behaviour of $L_k$ in (\ref{eq194}) for the auxiliary system of Fig.~\ref{fig6}. Moreover, suppose that $\eta$ has a variance equal to ${{\sigma}_{{\hat{\psi}}_G}^{2}}$. So based on the asymptotic equivalence between the matrix representations of $L$ and ${\{L_k\}}$, choosing $J$, $B$ and $\eta$ as above will render the system of Fig.~\ref{fig6} well-posed and internally stable. More specifically, the latter set of filters and the noise will give WSS processes to which $\hat{u}_G$ and $\hat{z}_G$ converge. Therefore, for the control input $u'$ and error signal $z'$ in the feedback loop of Fig.~\ref{fig6}, ${S}_{u'}={S}_{\breve{u}}$ and ${{\sigma}_{{\hat{z}}_{G}}^{2}}={{\sigma}_{{z'}}^{2}}$ hold. Then based on Lemma~\ref{lemma2}, the directed information rate in the NCS of Fig.~\ref{fig6} can be expressed as
\begin{IEEEeqnarray}{rcl}\label{eq196}
	\begin{split}
		{{I_{\infty}^{(h)}}(y'\to{u'})}=\frac{1}{4\pi}\int_{-\pi}^{\pi} \log\big(\frac{{S_{u'}}(e^{j\omega})}{{\sigma}_{\eta}^{2}}\big)d\omega=\frac{1}{4\pi}\int_{-\pi}^{\pi} \log\big(\frac{{S_{\breve{u}}}(e^{j\omega})}{{\sigma}_{{\hat{\psi}}_G}^{2}}\big)d\omega,
	\end{split}	
\end{IEEEeqnarray}
First, we can deduce that any pair $(\hat{{E}},\hat{{D}})$ with properties stated above has a counterpart comprised of LTI filter $B$, $J=1$ and the white Gaussian noise $\eta$ in architecture of Fig.~\ref{fig6} in such a way that 
\begin{figure}[thpb]
	\centering
	\includegraphics[width=8cm,height=5cm]{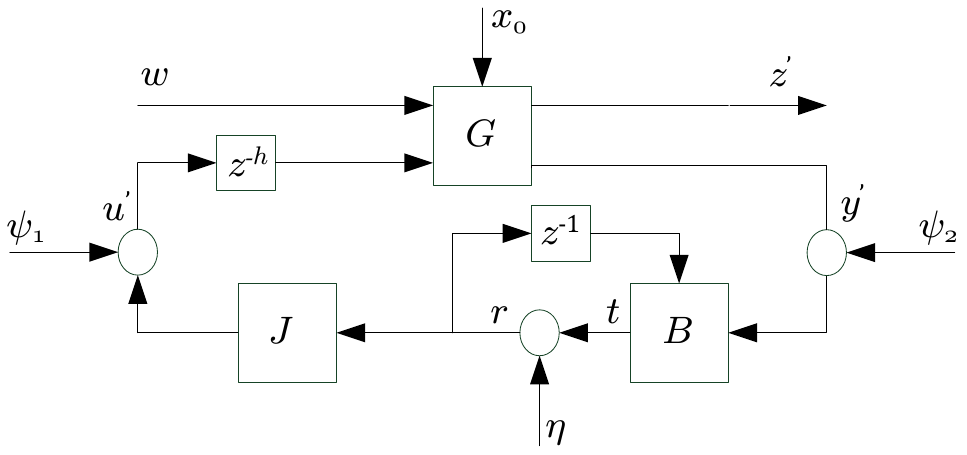}
	\caption{The LTI system whose internal stability guarantees the internal stability of the auxiliary system in Fig.~\ref{fig6}}
	\label{fig7}
\end{figure}
${{I_{\infty}}(y'\to{u'})}\leq{{I_{\infty}}(\hat{y}\to\hat{u})}$ and ${{\sigma}_{\hat{z}}^{2}}={{\sigma}_{{z'}}^{2}}$. Secondly, the main problem is finding the infimum of $\mathcal{R}$ over all mappings $(\hat{{E}},\hat{{D}})$. With all of this in mind, it can be implied from (\ref{eq196}) and (\ref{eq192}) that the lower bound for $\mathcal{R}(D)$ would be equal to the rightmost term of (\ref{eq191}) which completes the proof.
\subsection{Proof of Lemma~\ref{lemma3}}\label{app0lemma3}
The necessary and sufficient condition for the feedback loop of Fig.~\ref{fig6} to be internally stable and well-posed is that every entry of the transfer function matrix from input ${[\eta,w,{\psi}_1,{\psi}_2]^T}$ to outputs ${[z',y',r,u']^T}$ in the system of Fig.~\ref{fig7} belongs to ${\mathcal{R}\mathcal{H}}_{\infty}$ \cite{francis1987course}. Such a transfer function matrix, which we denote by $T$, is described as follows:
 	\begin{equation}\label{eq27}
 	T=
 	\begin{bmatrix}
 	G_{12}J{z^{-h}}M&G_{11}+G_{12}J{z^{-h}}{B_y}MG_{21}&G_{12}{z^{-h}}(1-{B_{r}}{z^{-1}})M&{G_{12}}J{z^{-h}}{B_y}M\\
 	G_{22}J{z^{-h}}M&G_{21}(1-{B_{r}}{z^{-1}})M&G_{22}{z^{-h}}(1-{B_{r}}{z^{-1}})M&{G_{22}}J{z^{-h}}{B_y}M\\
 	M&{G_{21}}{B_y}M&{G_{22}}{z^{-h}}{B_y}M&{B_y}M\\
 	JM&{G_{21}}J{B_y}M&(1-{B_{r}}{z^{-1}})M&J{B_y}M
 	\end{bmatrix},
 	\end{equation}	
 where
\begin{equation}\label{eq40}
M\triangleq{(1-{B_{r}}{z^{-1}}-{G_{22}}J{z^{-h}}{B_y})}^{-1}.
\end{equation}
Now, let us shift the delay block in the system of Fig.~\ref{fig6} to the plant model in a way that for the newly obtained system, the plant is described by 
\begin{equation}\label{eq47}
{G_a}=\left[
\begin{array}{lr}
G_{11}&{z^{-h}}{G_{12}}\\
{G_{21}}&{z^{-h}}{G_{22}}
\end{array}\right]. 
\end{equation}
\noindent
\begin{figure}[thpb]
	\centering
	\includegraphics[width=8cm]{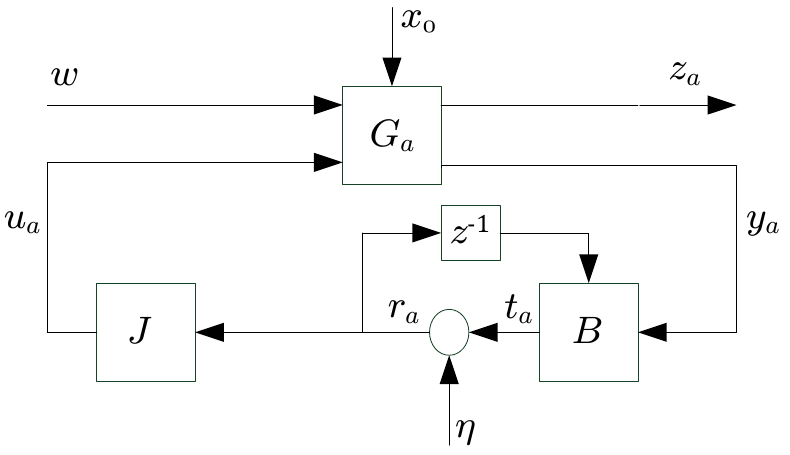}
	\caption{The equivalent system with the same ${{\varphi'}(D)}$ as the NCS of Fig.~\ref{fig6}}
	\label{fig10}
\end{figure}
\begin{figure}[thpb]
	\centering
	\includegraphics[width=8cm,height=5cm]{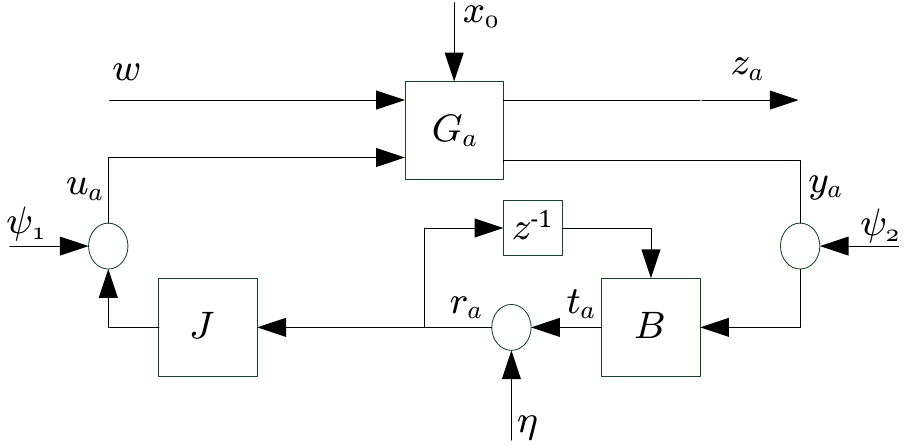}
	\caption{The auxiliary feedback loop characterizing the internal stability of the NCS of Fig.~\ref{fig10}}
	\label{fig11}
\end{figure}
Such an auxiliary NCS is depicted by Fig.~\ref{fig10}. Except for the plant model (\ref{eq47}), everything in the feedback loop of Fig.~\ref{fig10} is assumed to be the same as in the system of Fig.~\ref{fig6}. 
The internal stability and well-posed-ness of the feedback loop of Fig.~\ref{fig10} is guaranteed if and only if every entry of the transfer-function matrix, say ${T_a}$, from ${[\eta,w,{\psi}_1,{\psi}_2]^T}$ to ${[{z_a},{y_a},{r_a},{u_a}]^T}$ in Fig.~\ref{fig11} belongs to ${\mathcal{R}\mathcal{H}}_{\infty}$. It is straightforward to see that ${T_a}=T$. So an equivalence holds between internal stability and well-posed-ness of the system of Fig.~\ref{fig10} and the NCS of Fig.~\ref{fig6}. In other words, every triplet $(B,J,{{\sigma}_{\eta}^{2}})$ rendering the feedback loop of Fig.~\ref{fig10} internally stable and well-posed, will bring internal stability and well-posed-ness to the NCS of Fig.~\ref{fig6} as well. One other implication of ${T_a}=T$ is that using an stabilizing $(B,J,{{\sigma}_{\eta}^{2}})$ commonly for NCSs of Fig.~\ref{fig6} and Fig.~\ref{fig10} will lead to an identical ${\vartheta}_{r}^{'}({B},{J},{{\sigma}_{{\eta}}^{2}})$. 
This is due to the properties of LTI systems exposed to Gaussian and stationary inputs. Furthermore, those properties lead to deriving the following $H_2$-norm expressions for SNR and variance of the output $z'$ in the NCS of Fig.~\ref{fig6}:
\begin{IEEEeqnarray}{rl}\label{eq28}
	\begin{split}
		{\frac{{\sigma}_{t}^{2}}{{\sigma}_{\eta}^{2}}}&=\norm{{M-1}}_{2}^{2}+{\norm{{B_y}M{G_{21}}}_{2}^{2}}{{\sigma}_{\eta}^{-2}},\\
		{{\sigma}_{z'}^{2}}&=\norm{{{G_{11}}+{G_{12}}N{(1-{G_{22}}N)}^{-1}G_{21}}}_{2}^{2}+{\norm{{G_{12}}JM}_{2}^{2}}{{\sigma}_{\eta}^{2}},
	\end{split}
\end{IEEEeqnarray} 
in which $N\triangleq{J{B_y}{z^{-h}}{(1-{B_{r}}{z^{-1}})}^{-1}}$. Likewise, the SNR and variance of the output $z$ in the NCS of Fig.~\ref{fig10} is formalized in terms of $H_2$-norms as follows:
 
\begin{IEEEeqnarray}{rl}\label{eq29}
	\begin{split}
		{\frac{{\sigma}_{t_a}^{2}}{{\sigma}_{\eta}^{2}}}&=\norm{{{M_a}-1}}_{2}^{2}+{\norm{{B_y}{M_a}{G_{21}}}_{2}^{2}}{{\sigma}_{\eta}^{-2}},\\
		{{\sigma}_{z_a}^{2}}&=\norm{{{G_{11}}+{G_{12}}{z^{-h}}{N_a}{(1-{G_{22}}{z^{-h}}{N_a})}^{-1}G_{21}}}_{2}^{2}+{\norm{{G_{12}}J{M_a}}_{2}^{2}}{{\sigma}_{\eta}^{2}},
	\end{split}
\end{IEEEeqnarray}
\\
where ${M_a}=M$ and ${N_a}\triangleq{J{B_y}{(1-{B_{r}}{z^{-1}})}^{-1}}$. It follows from (\ref{eq28}) and (\ref{eq29}) that $({{{\sigma}_{t}^{2}}/{{\sigma}_{\eta}^{2}}})=({{{\sigma}_{t_a}^{2}}/{{\sigma}_{\eta}^{2}}})$ and ${{\sigma}_{z'}^{2}}={{\sigma}_{z_a}^{2}}$. Therefore, upon using the same stabilizing triplet $({B},{J},{{\sigma}_{{\eta}}^{2}})$, the channel SNR and the variance of the output characterizing performance will be the same for the NCSs of Fig.~\ref{fig6} and Fig.~\ref{fig10}.

According to \cite[Lemma~4.1]{silva2016characterization}, for any pair $(B,J)=({B_{1}},{J_{1}})$ that renders the feedback loop of Fig.~\ref{fig10} internally stable and well-posed, there exists another pair with the same properties as for $({B_2},{J_2})$ in this lemma. Then our claims follow immediately from the above equivalences between the NCS of Fig.~\ref{fig10} and the NCS of Fig.~\ref{fig6}.
\subsection{Proof of Corollary~\ref{cor1}}\label{app0cor1}
	The feasibility of obtaining ${{{\vartheta}'}_u}(D)$, caused by $D$ belonging to $({D_{\inf}}(h),\infty)$, certifies the existence of a triplet, say $({{B}_{\zeta}},1,{{\sigma}_{{\eta}_{\zeta}}^{2}})$, that leads to ${{\sigma}_{z'}^{2}}\leq{D}$ for the system of Fig.~\ref{fig6}. In the latter triplet, ${{B}_{\zeta}}$ is assumed to be a proper LTI filter and ${{\sigma}_{{\eta}_{\zeta}}^{2}}\in{\mathbb{R}^{+}}$. This together with the definition of ${{{\vartheta}'}_u}$ and ${{{\vartheta}'}_r}$ in (\ref{eq191}) and (\ref{eq197}), respectively, yields the following:
	\begin{equation}\label{eq231}
	{{{\vartheta}'}_u}(D)+{\zeta}\geq{{\vartheta}_{r}^{'}({{B}_{\zeta}},1,{{\sigma}_{{\eta}_{\zeta}}^{2}})}, \forall\zeta\in{\mathbb{R}^{+}}.
	\end{equation}
	\\ 
	Moreover, the triplet $({{B}_{\zeta}},1,{{\sigma}_{{\eta}_{\zeta}}^{2}})$ with aforementioned properties meets the conditions in Lemma~\ref{lemma3}. Therefore, another triplet, say $({{\tilde{B}}_{\zeta}},{{\tilde{J}}_{\zeta}},{{\sigma}_{{\eta}_{\zeta}}^{2}})$, exists in such a way that implementing it brings internal stability and well-posed-ness, keeps ${\sigma}_{z'}^{2}$ intact, and yields
	\begin{equation}\label{eq232}
	{{{\vartheta}'}_u}(D)+{\zeta}\geq{{{\frac{1}{2}}\log(1+{\frac{{\sigma}_{t}^{2}}{{\sigma}_{\eta}^{2}}}){\mid}_{(B,J,{\sigma}_{\eta}^{2})=({\tilde{B}}_{\zeta},{\tilde{J}}_{\zeta},{{\sigma}_{{\eta}_{\zeta}}^{2}})}}-{\rho}} 
	\end{equation}
	for the LTI feedback loop of Fig.~\ref{fig6}. Note that ${{\tilde{J}}_{\zeta}}$ is a biproper filter while ${{\tilde{B}}_{\zeta}}$ only needs to be proper. Now the fact that (\ref{eq232}) holds for any $\zeta,\rho>0$, the definition of ${{\varphi'}(D)}$ in (\ref{eq23}), and the claim of Theorem~\ref{th3} complete the proof. 
\subsection{Proof of Lemma~\ref{lemma4}}\label{app0lemma4}
Let $\tilde{G}_h$ denote the transfer-function matrix from ${[\tilde{w}\quad r]^T}$ to ${[\tilde{z}\quad t]^T}$ in Fig.~\ref{fig8}. Since $r_h$ is related to $r$ by $r_h=r{z}^{-h}$, we can conclude that $\tilde{G}_h$ meets the conditions of being proper and real rational, and containing a strictly proper SISO open-loop transfer function from $r$ to $t$. Now having the schemes described via (\ref{eq38}) and (\ref{eq39}) in mind, we can deduce our claim immediately from \cite[Lemma~5.1]{silva2016characterization}.	
\subsection{Proof of Lemma~\ref{lemma5}}\label{app0lemma5}
Let us assume that a linear source coding scheme is implemented in the feedback path of the main system in Fig.~\ref{fig1}. Due to the feasibility of finding $\varphi'(D)$, which necessitates satisfaction of Assumption~\ref{ass21}, we can conclude the existence of proper LTI filters $B$ and $J$ that together with an AWGN, say $\eta$, render the NCS of Fig.~\ref{fig1} SAWSS. It stems from some properties of internal stability that the system will still be stable if one keeps the latter filters $B$ and $J$ and only sets $\eta$ as $\eta=0$. This signifies that in the case of unity feedback ($t=r$), internal stability and well-posed-ness are guaranteed for the open-loop system between $r$ and $t$. We come immediately to the conclusion that (\ref{eq41}) holds based on \cite[Corollary~5.3]{silva2011framework} and statistical characteristics of the dither mentioned in Lemma~\ref{lemma4}.
\subsection{Proof of Theorem~\ref{th4}}\label{app0th4}
Considering the feasibility of finding ${{\varphi'}(D)}$, results of Lemma~\ref{lemma3}, lemma~\ref{lemma4}, and Lemma~\ref{lemma5}, and  invertibility of the decoder, we conclude the claim  by following the same steps as in \cite[Theorem~5.1]{silva2016characterization}. 
\subsection{Proof of Lemma~\ref{lemma01}}\label{app0lemma01}
 One of the common feedback loop components across the considered cases in Fig.~\ref{fig14} is the LTI plant $G$ which is described by state-space difference equations as follows:
\begin{equation}\label{eq52}
G:
\begin{cases}
{x(k+1)}=Ax(k)+{B_1}w(k)+{B_2}u(k)\\
z(k)={C_1}x(k)+{D_{11}}w(k)+{D_{12}}u(k)\\
y(k)={C_2}x(k)+{D_{21}}w(k),
\end{cases}
\end{equation} 
where $x\in{{\mathbb{R}}^{n_x}}$ represents plant states and $u$, $w$, $y$ and $z$ are inputs and outputs defined as in (\ref{eq1}). Moreover, $A$, $B_1$, $B_2$, $C_1$, $C_2$, $D_{11}$, $D_{12}$, and $D_{21}$ are time-invariant matrices of appropriate dimensions. According to the recursion in (\ref{eq52}), the states and outputs of the plant at each time instant $i\in{\mathbb{N}_0}$ can be expressed in terms of initial conditions, disturbance and control inputs as follows:
\begin{equation}\label{eq53}
\begin{cases}
{x(i)}={A^i}x(0)+{{\mathcal{B}}_1}(i){w^{i-1}}+{{\mathcal{B}}_2}(i){u^{i-1}}\\
z(i)={C_1}{A^i}x(0)+{{\mathcal{D}}_{11}(i)}{w^{i}}+{{\mathcal{D}}_{12}(i)}{u^{i}}\\
y(i)={C_2}{A^i}x(0)+{{\mathcal{D}}_{21}(i)}{w^{i}}+{{\mathcal{D}}_{22}(i)}{u^{i-1}},
\end{cases}
\end{equation}
where the involved matrices are defined as
\begin{align}\label{eq54}
\begin{split}
&{{\mathcal{B}}_1}(i)=[{A^{i-1}}{B_1}~{A^{i-2}}{B_1}\dots{B_1}]\\
&{{\mathcal{B}}_2}(i)=[{A^{i-1}}{B_2}~{A^{i-2}}{B_2}\dots{B_2}]\\
&{{\mathcal{D}}_{11}(i)}=[{{C_1}A^{i-1}}{B_1}~{C_1}{A^{i-2}}{B_1}\dots{C_1}{B_1}~D_{11}]\\
&{{\mathcal{D}}_{12}(i)}=[{{C_1}A^{i-1}}{B_2}~{C_1}{A^{i-2}}{B_2}\dots{C_1}{B_2}~D_{12}]\\
&{{\mathcal{D}}_{21}(i)}=[{{C_2}A^{i-1}}{B_1}~{C_2}{A^{i-2}}{B_1}\dots{C_2}{B_1}~D_{21}]\\
&{{\mathcal{D}}_{22}(i)}=[{{C_2}A^{i-1}}{B_2}~{C_2}{A^{i-2}}{B_2}\dots{C_2}{B_2}].\\
\end{split}
\end{align} 
For the case where the time delay is imposed by the error-free digital channel between the encoder-controller and the decoder-controller, the relationship between the control input and the sensor output is characterized based on (\ref{eq24})-(\ref{eq271}). The dynamics described by  (\ref{eq24})-(\ref{eq271}) can be summarizd in the constant channel delay case as follows:
\begin{align}\label{eq2}
\begin{split}
{y_q(k)}&={E_k}({y^k},{\eta_{e}^{k}})\\
{{{u}_{q}}(k)}&={{{y}_{q}}(k-h)}\\
{u(k)}&={D_k}({{u}_{q}^{k}},{\eta_{d}^{k}}),
\end{split}
\end{align}
where $E_k$ and $D_k$ represent causal, but otherwise arbitrary, mappings at  each $k\in{\mathbb{N}_0}$.
It follows from (\ref{eq2}) that ${u^k}$ can be stated as an arbitrary function, say ${N_k}$, of $({\eta_{d}^{k}},{y^{k-h}},{\eta_{e}^{k-h}})$, i.e.,   ${u^k}={N_k}({\eta_{d}^{k}},{y^{k-h}},{\eta_{e}^{k-h}})$. Then from (\ref{eq53}) and by induction, we can conclude that at each time instant $k\in{\mathbb{N}_0}$, $x(k)$ is a function of $(x(0),{w^{k-1}},{{\eta}_{d}^{k-1}},{{\eta}_{e}^{k-1-h}})$, $z(k)$ is a function of $(x(0),{w^{k}},{{\eta}_{d}^{k}},{{\eta}_{e}^{k-h}})$, and $y(k)$ is a function of $(x(0),{w^{k}},{{\eta}_{d}^{k-1}},{{\eta}_{e}^{k-1-h}})$.       

In the second case, it is the link between the decoder-controller and the plant that induces the time delay. For such a setting, ${E_k}$, ${{D}_k}$, ${{\eta_{e}}(k)}$ and ${\eta_{d}(k)}$ yield a scheme with following dynamics:  
\begin{align}\label{eq58}
\begin{split}
{y_q(k)}&={E_k}({y^k},{\eta_{e}^{k}})\\
{{{u}_{q}}(k)}&={{{y}_{q}}(k)}\\
{{u}(k)}&={D_{k-h}}({{u}_{q}^{k-h}},{\eta_{d}^{k-h}}).
\end{split}
\end{align}   
It follows from $(\ref{eq58})$ that in this case, $u^k$ can be expressed as ${u^k}={M_k}({y^{k-h}},{{\eta_d}^{k-h}},{{\eta}_{e}^{k-h}})$, $\forall{k}\in{\mathbb{N}_0}$, where $M_k$ is an arbitrary mapping which is specified by $\{E_i\}_{i=0}^{k-h}$ and $\{D_i\}_{i=0}^{k-h}$. Substituting such an expression into $(\ref{eq53})$ an by induction, we can rederive $x(k)$, $z(k)$, and $y(k)$ as functions of $(x(0),{w^{k-1}},{{\eta}_{d}^{k-h-1}},{{\eta}_{e}^{k-1-h}})$, $(x(0),{w^{k}},{{\eta}_{d}^{k-h}},{{\eta}_{e}^{k-h}})$, and $(x(0),{w^{k}},{{\eta}_{d}^{k-h-1}},{{\eta}_{e}^{k-h-1}})$, respectively.  

As the third case, we focus on a structure in which the delay is introduced by the path between the sensor and the encoder-controller. In this situation, the coding scheme is described by causal mappings ${E_k}$ and ${{D}_k}$, and side informations ${{\eta_{e}}(k)}$ and ${\eta_{d}(k)}$, as follows:
\begin{align}\label{eq10}
\begin{split}
{y_q(k)}&={E_k}({y^{k-h}},{\eta_{e}^{k}})\\
{{{u}_{q}}(k)}&={{{y}_{q}}(k)}\\
{{u}(k)}&={D_k}({{u}_{q}^{k}},{\eta_{d}^{k}}).
\end{split}
\end{align} 
Taking the same steps as for the previous cases, we derive ${u^k}={S_k}({y^{k-h}},{{\eta_d}^{k}},{{\eta}_{e}^{k}})$, $\forall{k}\in{\mathbb{N}_0}$, where $S_k$ is a causal mapping and a function of $\{E_i\}_{i=0}^{k}$ and $\{D_i\}_{i=0}^{k}$. Then considering (\ref{eq53}) and based on induction, we come to the conclusion that for the closed-loop system considered in this case, $x(k)$ is a function of $(x(0),{w^{k-1}},{{\eta}_{d}^{k-1}},{{\eta}_{e}^{k-1}})$, $z(k)$ is a function of $(x(0),{w^{k}},{{\eta}_{d}^{k}},{{\eta}_{e}^{k}})$, and $y(k)$ is a function of $(x(0),{w^{k}},{{\eta}_{d}^{k-1}},{{\eta}_{e}^{k-1}})$ for all $k\in{\mathbb{N}_0}$.

According to the above observations, comparing system states $x$, sensor output $y$, and the output $z$ at each time instant indicates that such signals are not necessarily equal across the three cases studied above if the systems share the design (mappings for coding and control and side information) and have the same initial conditions and exogenous inputs. So values of each signal change by relocating the delay component in the NCS of Fig.~\ref{fig1}. However, it is straightforward to see from the structure of the variables describing processes $x$, $z$, and $y$ that the equivalence over cases can be obtained under the condition that everything is the same across the cases except for side information which can be considered as decision variable. 
\\
\bibliographystyle{IEEEtran}
\bibliography{refs}
%

%
%
%




\end{document}